\documentclass[final,3p,times,twocolumn]{elsarticle}

\usepackage{amssymb}

\usepackage[T1]{fontenc}
\usepackage{amsfonts}
\usepackage{amsmath}
\usepackage{amsthm}
\usepackage{bbm}
\usepackage{xcolor}
\usepackage{hyperref}
\hypersetup{colorlinks, linkcolor=red, citecolor=blue, urlcolor=blue}
\newcommand{\EE}{\mathbb{E}}
\newcommand{\PP}{\mathbb{P}}
\newcommand{\RR}{\mathbb{R}}

\newcommand{\cA}{\mathcal{A}}
\newcommand{\cE}{\mathcal{E}}
\newcommand{\cI}{\mathcal{I}}

\newcommand{\cN}{\mathcal{N}}

\newcommand{\cS}{\mathcal{S}}
\newcommand{\cZ}{\mathcal{Z}}

\DeclareMathOperator*{\argmax}{arg\,max}

\DeclareMathOperator{\erf}{erf}

\let\eps\varepsilon
\let\to\longrightarrow
\newcommand{\indic}{\mathbbm{1}}

\newcommand{\conv}[2][\ ]{\overset{#1}{\underset{#2}{\to}}}

\newcommand{\aseq}[1]{\underset{#1}{=}}
\newcommand{\equi}[1]{\underset{#1}{\sim}}
\newcommand{\mbf}[1]{\mathbf{#1}}

\theoremstyle{definition}
\newtheorem{assu}{Assumption}

\theoremstyle{plain}
\newtheorem{prop}{Proposition}
\newtheorem{lem}{Lemma}

\journal{Probabilistic Engineering Mechanics}

\begin{document}
\begin{frontmatter}

\title{Reference prior for Bayesian estimation of seismic fragility curves}
\author[label1,label2]{Antoine Van Biesbroeck\corref{cor1}}
\ead{antoine.van-biesbroeck@polytechnique.edu}
\author[label3]{Clément Gauchy}
\author[label2]{Cyril Feau}
\author[label1]{Josselin Garnier}

\affiliation[label1]{organization={CMAP, CNRS, École polytechnique, Institut Polytechnique de Paris},%
            city={91120 Palaiseau},
            country={France}}

\affiliation[label2]{organization={Université Paris-Saclay, CEA, Service d'Études M\'ecaniques et Thermiques},%
            city={91191 Gif-sur-Yvette},
            country={France}}

\affiliation[label3]{organization={Université Paris-Saclay, CEA, Service de Génie Logiciel pour la Simulation},%
            city={91191 Gif-sur-Yvette},
            country={France}}

\cortext[cor1]{Corresponding author}

\begin{abstract}
      {{One of the key elements of probabilistic seismic risk assessment studies is the fragility curve, which represents the conditional probability of failure of a mechanical structure for a given scalar measure derived from seismic ground motion. For many structures of interest, estimating these curves is a daunting task because of the limited amount of data available; data which is only binary in our framework, i.e., only describing the structure as being in a failure or non-failure state. A large number of methods described in the literature tackle this challenging framework through parametric log-normal models. Bayesian approaches, on the other hand, allow model parameters to be learned more efficiently. However, the impact of the choice of the prior distribution on the posterior distribution cannot be readily neglected and, consequently, neither can its impact on any resulting estimation. This paper proposes a comprehensive study of this parametric Bayesian estimation problem for limited and binary data. Using the reference prior theory as a cornerstone, this study develops an objective approach to choosing the prior. This approach leads to the Jeffreys prior, which is derived for this problem for the first time. The posterior distribution is proven to be proper (i.e., it integrates to unity) with the Jeffreys prior but improper with some traditional priors found in the literature. With the Jeffreys prior, the posterior distribution is also shown to vanish at the boundaries of the parameters' domain, which means that sampling the posterior distribution of the parameters does not result in anomalously small or large values. Therefore, the use of the Jeffreys prior does not result in degenerate fragility curves such as unit-step functions, and leads to more robust credibility intervals. The numerical results obtained from two different case studies---including an industrial example---illustrate the theoretical predictions.}
      }
\end{abstract}


\end{frontmatter}

\section{Introduction}

Fragility curves are key assets for the probabilistic seismic risk assessment (SPRA) of mechanical structures. Introduced in the 1980s for seismic risk assessment studies carried out on nuclear facilities (for example, \cite{Kennedy1980,KENNEDY198447, PARK1998, KENNEDY1999, Cornell2004}), they express the probability of failure of mechanical structures as a function of a scalar value derived from seismic ground motions, called the intensity measure (IM). One example of such a scalar value is the peak ground acceleration (PGA). In \cite{Cornell2004}, \citeauthor{Cornell2004} states the condition under which reducing seismic risk to IM values is relevant. This condition is the so-called ``sufficiency assumption'' of the IM with regards to the magnitude M, the source-site distance R, and other parameters which are thought to drive seismic risk at the location considered \cite{Grigoriu2021}.
Fragility curves can be estimated using data collated from various sources, namely: (i) expert assessments supported by test data \cite{Kennedy1980, KENNEDY198447, PARK1998, Zentner2017}, (ii) experimental data \cite{PARK1998, Gardoni2002, Choe2007}, (iii) damage reports---called empirical data---obtained from existing structures that have been subjected to an earthquake \cite{Shinozuka2000, Lallemant2015, Straub2008}, and (iv) analytical results obtained from various numerical models using artificial or natural seismic excitations \cite{ZENTNER20101614, WangF2020,MANDAL201611,WANGZ2018,WANG2018232,ZHAO2020103}.

Parametric fragility curves were historically introduced in the SPRA framework because an estimation of such curves is possible with small sample sizes {and binary data, i.e., when the only indication is the state of the structure: failure or non-failure. Since then, the log-normal model has become the most widely used \cite{Shinozuka2000, Lallemant2015, Straub2008, ZENTNER20101614, MANDAL201611, WANGZ2018, WANG2018232, ZHAO2020103, ELLINGWOOD2001251, KIM2004, Mai2017, TREVLOPOULOS2019}, and it remains prevalent to this day due to its proven capability to handle limited and binary data (e.g., \cite{WangF2020, Katayama2021, KHANSEFID2023, LEE2023}).}
In practice, several strategies can be implemented to adjust the two parameters for this model, namely the median $\alpha$ and the log standard deviation $\beta$.
When the data is binary, \citeauthor{Lallemant2015} \cite{Lallemant2015} recommends a strategy called the maximum likelihood estimation (MLE). This technique is one of the most widely used. When the data samples are independent from each other, the bootstrap technique can be used as an additional measure in order to obtain confidence intervals related to the size of the samples considered \cite{Shinozuka2000, ZENTNER20101614, WangF2020}. 

If the data set contains more information than a simple indication of failure, i.e., if data isn’t binary, techniques based on machine learning can also be exploited. This is the case, for instance, when a mechanical structural failure is the result of an engineering demand parameter (EDP) exceeding a threshold limit, said EDP being observed as part of an experiment (numerical or practical). Examples of such techniques are: linear regression or generalized linear regression \cite{Lallemant2015}, classification-based techniques \cite{BERNIER2019, KIANI2019, Sainct2020}, kriging \cite{Gidaris15, GENTILE2020101980, Gauchy2022}, polynomial chaos expansion \cite{Mai16}, stochastic polynomial chaos expansions \cite{ZHU2023}, and artificial neural networks (ANN) \cite{WANG2018232, MITROPOULOU2011, WANGZ2018}. Whenever data is obtained through numerical simulations, some of these methods can be coupled with adaptive techniques to reduce the number of calculations required \cite{WANG2018232, Sainct2020, Gidaris15, Gauchy2021}. Some of these techniques are compared and their strengths and weaknesses highlighted in \cite{Lallemant2015}.

The Bayesian framework has recently become increasingly popular in seismic fragility analysis \cite{Gardoni2002, WANG2018232, Katayama2021, LEE2023, KOUTSOURELAKIS2010, damblin2014, TADINADA201749, KWAG20181, Jeon2019, TABANDEH2020}. It actually allows solving the irregularity issues encountered when estimating parametric fragility curves by using, for example, the MLE method, which can indeed lead to unrealistic or degenerate fragility curves such as unit-step functions when only limited data is available.
Such problems can particularly be encountered when resorting to high-fidelity models (i.e., complex and detailed) because of the calculation load or for example when experimental tests are carried out on shaking tables. In earthquake engineering, Bayesian inference is often used to update existing log-normal fragility curves previously obtained through various approaches, assuming independent distributions for the prior values of $\alpha$ and $\beta$, such as log-normal distributions. 
For example, in \cite{TADINADA201749} and \cite{KWAG20181}, the median prior values come from equivalent linearized mechanical models. In \cite{WANG2018232}, both aleatory and epistemic uncertainties are taken into account in the parametric model originally introduced in \cite{Kennedy1980}: an ANN is trained and used to characterize (i) the aleatory uncertainty and (ii) the prior median value of $\alpha$, while the associated epistemic uncertainty is taken from the existing literature. The log-normal prior distribution of $\alpha$ is then updated with empirical data. 
In \cite{Katayama2021}, the results of incremental dynamic analysis are used to obtain a prior value of $\alpha$, whereas the prior value of $\beta$ is determined through a parametric study. This results in satisfactory convergence, whatever its target value, before application to practical problems.
In \cite{Straub2008}, \citeauthor{Straub2008} mainly focus on the implications for fragility analyses of statistical dependencies within the data. The prior is defined as the product of a normal distribution for $\ln(\alpha)$, and the improper distribution $1/\beta$ for $\beta$. The definition of the normal distribution is based on engineering assessments, assuming that, for the relevant component for example, the median lies between 0.02~g and 3~g with a probability of 90\%. This prior was preferred to $1/\alpha$ on the grounds that it led to unrealistically large posterior values of $\alpha$. A sensitivity analysis is further performed to examine the impact of the choice of the prior distribution on the final results. The Bayesian framework is also relevant for fitting numerical models (e.g., mathematical expressions based on engineering assessments or physics-based models) to experimental data in order to estimate fragility curves \cite{Gardoni2002, TABANDEH2020} or metamodels such as logistic regressions \cite{KOUTSOURELAKIS2010, Jeon2019}.

{In this paper, we will deal with limited sets of binary data and will consider the log-normal model in a Bayesian framework. In this sense, this paper will mainly address equipment problems for which only binary results of seismic qualification tests (e.g., tests of electrical relays, etc.) or empirical data such as presented in \cite{Straub2008} are available. However, the methodology developed here could perfectly be applied to simulation-based approaches as well.} The Bayesian perspective focuses on the impact of the prior on the estimations of parametric fragility curves, as part of the SPRA framework. With a limited data set, the impact of the choice of the prior on the posterior distribution cannot be neglected and, consequently, neither does its impact on the estimation of any key asset related to the fragility curves. In this study, the goal is to choose the prior while eliminating, insofar as it is possible, any subjectivity which would unavoidably lead to open questions regarding the impact of the prior on the final results. The reference prior theory defines relevant metrics for determining whether a prior can be called ``objective'' \cite{Kass1996}. This allows us to focus on the well-known Jeffreys prior, the asymptotic optimum of the ``mutual information" w.r.t. the size of the data set \cite{Berger2009}, and which will be explicitly derived for the first time in this study. Of course, from a subjectivity perspective, the choice of a parametric model for the fragility curve is debatable. However, numerical experiments based on the seismic responses of mechanical systems suggest that the choice of an appropriate IM makes it possible to reduce the potential biases between reference fragility curves (that can be obtained by massive Monte-Carlo methods) and their log-normal estimations \cite{Gauchy2021}. This observation is reinforced by recent studies on the impact of IMs on fragility curves \cite{Sainct2020, CIANO2020, CIANO2022}. {In this paper, we will ensure the relevance of the estimations by comparing them to the results of massive Monte-Carlo methods on academic examples. Although the numerical results are illustrated with the PGA, the proposed methodology is independent of the choice of the IM, and it can be implemented with any IM of interest, without additional complexity.}

After formulating the problem from a Bayesian perspective in the next section, we will review the objective prior theory in Section~\ref{sec:objprior}. Our main contributions begin in Section~\ref{sec:construction}, where the reference prior is explicitly derived. In Section~\ref{sec:tools}, we will present the estimation tools and the performance evaluation metrics used for this paper. These are implemented in Section~\ref{sec:application} on {two} different case studies. For each case, the \emph{a posteriori} distributions of the parameters of the log-linear probit model and of the corresponding fragility curves are compared with those obtained with traditional \emph{a priori} choices found in the literature. Section~\ref{sec:conclusion} will conclude this paper and summarize our findings. \ref{app:asymptotics} and \ref{app:SKreview} deal with mathematical results regarding the asymptotic properties of the priors and posteriors considered in this work. In particular, \ref{app:lik-degenarcy-discuss} explains the apparition of degenerate and unrealistic fragility curves with the MLE or Bayesian estimations of traditional priors.

\section{{Bayesian model for parametric log-normal seismic fragility curves}}\label{sec:pb}

As mentioned in the introduction, a log-linear probit model is often used to estimate fragility curves. In this model, the probability of failure given the IM takes the following form: 
        \begin{equation}\label{eq:Pfa}
            P_f(a)=\PP(\mbox{`failure'}|\text{IM}=a) = \Phi\left(\frac{\log a-\log\alpha}{\beta}\right) ,
        \end{equation}
    where $\alpha, \beta\in (0,+\infty)$ are the two model parameters and $\Phi$ is the cumulative distribution function of a standard Gaussian variable. In the following, we denote $\theta=(\alpha,\beta)$. 

    From a Bayesian perspective, $\theta$ is considered as a random variable \cite{Robert2007}. Its distribution is denoted by $\pi$ and called the prior. We denote by $\Theta\subset [0,+\infty)^2$ the support of the prior distribution.
    
    Our statistical model consists of the observations of independent realizations $(a_1,z_1),\dots,(a_k,z_k)\in\cA\times \{0,1\}$, {where $\cA \subset [0,+\infty)$ is the support of the distribution of the IM and $k$ is the size of the data set}. For the $i$th seismic event, $a_i$ is the observed IM, and $z_i$ is the observation of a failure ($z_i$ is equal to one if failure has been observed during the $i$th seismic event, and is equal to zero otherwise). %
    The joint conditional distribution of the pair $(a,z)$ on $\theta$ has the form:
        \begin{equation}
        \label{eq:pazgiventheta}
            p(a,z|\theta) %
            = p(a)p(z|a,\theta) ,
        \end{equation}
    where $p(a)$ denotes the distribution of the IM and $p(z|a,\theta)$ is a Bernoulli distribution whose parameter (the probability of failure) depends on $a$ and $\theta$ as expressed by Eq.~(\ref{eq:Pfa}). The product of the conditional distributions $p(z_i|a_i,\theta)$ is the likelihood of our model expressed as:
        \begin{equation}\label{eq:likelihood}
            \ell_k(\mbf z|\mbf a,\theta) = \prod_{i=1}^k\Phi\left(\frac{\log\frac{a_i}{\alpha}}{\beta}\right)^{z_i}\left(1-\Phi\left(\frac{\log\frac{a_i}{\alpha}}{\beta}\right)\right)^{1-z_i}  ,
        \end{equation}
    denoting $\mbf a=(a_i)_{i=1}^k$, $\mbf z=(z_i)_{i=1}^k$.
    
    The \emph{a posteriori} distribution of $\theta$ can be computed by the Bayes theorem. The resulting distribution %
        \begin{equation}\label{eq:posterior}
            p(\theta|\mbf a,\mbf z)=\frac{\ell_k(\mbf z|\mbf a,\theta)\pi(\theta)}{\int_\Theta \ell_k(\mbf z|\mbf a,\theta')\pi(\theta') d\theta'}
        \end{equation}
        is called the posterior.
    Sampling $\theta$ with the posterior distribution allows for the estimation of any relevant quantity. This method is explained further in Section \ref{sec:BayesFram}. Note that the Bayesian method requires choosing the prior $\pi$. In the next section, we will discuss how to make such a choice without any subjectivity.

\section{Reference prior theory}\label{sec:objprior}

This section is devoted to the choice of the prior in a Bayesian context. To this end, we have to deal with the idea of mutual information. Its definition and its usefulness in Bayesian problems are not new \cite{Bernardo1979, Kass1996}. However, its usefulness in estimating seismic fragility curves has not yet been studied in the literature. Shannon's information theory provides relevant elements about this. Information entropy is a common example that helps distinguish between an informative and a non-informative distribution \cite{Jaynes1982}.

One way to define a non-subjective prior is to look for a non-informative one (i.e., one with high entropy). However, this type of prior leads to posterior distributions that are quite unaffected by the likelihood of the statistical model. This can lead to the relevant parameters taking unrealistic posterior values when a limited amount of data is available. The consequence of this is a weaker convergence of the resulting estimations. Moreover, in practice and in earthquake engineering in particular, it is difficult to fully define a prior using a ``rigorous" approach. For example, in the case of a given distribution (which is already a subjective choice and not always easy to justify), the median can be obtained beforehand through a less-refined mechanical model. There remains, however, the matter of choosing the associated variance, of which we have just said that the consequences for the convergence of the \emph{a posteriori} estimations cannot be neglected. For all these reasons, it is relevant to search for an \emph{a priori} with objective information. 

To choose such a prior, let us consider the criterion introduced by \citeauthor{Bernardo1979} \cite{Bernardo1979} to define the so-called reference priors. The idea is to select the prior $\pi$ that maximizes the mutual information indicator $I(\pi|k)$, which expresses the information provided by the data to the posterior, relatively to the prior. In other words, the purpose of this criterion is to identify the prior that maximizes the capacity to ``learn" from observations. The mutual information indicator is defined by:
    
    {
    \begin{equation}
        I(\pi|k)= \sum_{ {\mbf z} \in \{0,1\}^k}  \int_{\cA ^k} KL(p(\cdot|\mbf a,\mbf z)||\pi)p(\mbf a,\mbf z)\prod_{l=1}^k  da_l ,
        \label{eq:defI}
    \end{equation}
    where the posterior $p(\cdot|\mbf a,\mbf z)$ is given by (\ref{eq:posterior}), and the joint distribution $p(\mbf a,\mbf z)$ is from (\ref{eq:pazgiventheta}-\ref{eq:likelihood}):
    \begin{align}
    p(\mbf a,\mbf z) &= \int_{\Theta} \prod_{l=1}^k p(a_l,z_l|\theta) \pi(\theta) d\theta\\
    &= \int_{\Theta} \prod_{l=1}^k p(a_l)  \ell_k(\mbf z|\mbf a,\theta) \pi(\theta) d\theta .\nonumber
    \end{align}
As discussed in Section~\ref{sec:practseism}, $p(a)$ can be assumed to be a log-normal probability distribution function, derived from seismic signals not included in the observations. The indicator in Eq.~(\ref{eq:defI}) is therefore independent from the data set considered.} This indicator is based on the Kullback-Leibler divergence between the posterior and the prior, which is known to express the idea that the information is provided by one distribution to another:
    \begin{equation}
        KL(p||q)=\int_{\Theta} p(\theta)\log\frac{p(\theta)}{q(\theta)} d \theta.
        \label{eq:defKL}
    \end{equation}

According to the literature, a reference prior can be suitably defined as the result of an asymptotic optimization of this mutual information metric \cite{Berger2009, Clarke1994}. It has been proven that, under some mild assumptions which are satisfied in our framework, the Jeffreys prior defined by
       \begin{equation}
        J(\theta)\propto\sqrt{|\det\cI^k(\theta)|} ,
        \label{eq:jeff}
    \end{equation}
    is the reference prior, with $\cI^k$ being the Fisher information matrix:
     {
   \begin{align}
        &\cI(\theta)^k_{i,j}  \\
           & = -\sum_{\mbf z\in \{0,1\}^k} \int_{\cA^k} \ell_k(\mbf z|\mbf a,\theta)\frac{\partial^2 \log \ell_k(\mbf z|\mbf a,\theta)}{\partial\theta_i\partial \theta_j}\prod_{l=1}^k p(a_l)  da_l .\nonumber
    \end{align}   
    }
    The property $\cI(\theta)^k=k\cI(\theta)$ makes $J$ independent of $k$, since it is defined up to a multiplicative constant. The Jeffreys prior is already well known in Bayesian theory for being invariant by a re-parametrization of the statistical model \cite{Bernardo2005}. This property is essential as it makes the choice of the model parameters $\theta$ without any incidence on the resulting posterior.
    
\section{Constructing the Jeffreys prior} \label{sec:construction}

    Based on the elements discussed in the previous section, the Jeffreys prior seems to be the best objective candidate for this problem. In this section, we will therefore calculate the Jeffreys prior in order to estimate log-normal seismic fragility curves with binary data. To our knowledge, the application of the reference prior theory to this field of study is completely new. The explicit calculation of this prior is carried out in Section~\ref{sec:jeffcalc}. It is followed in Section~\ref{sec:practseism} by an explanation about the practical implementation suggested and discussed in Section~\ref{sec:JeffDiscussion}. That last section in particular tackles the question of the proper characteristic of its resulting posterior, which is essential for the validation of any MCMC-based posterior sampling algorithm.
    
    \subsection{Calculating the Jeffreys prior}\label{sec:jeffcalc}
    
         The first step consists in computing the Fisher information matrix $\cI(\theta)=\cI(\theta)^1$ in our model, defined in Eq.~\eqref{eq:likelihood}. 
        Here, $\theta=(\alpha,\beta)\in \Theta$ and 
        \begin{equation}
            \cI(\theta)_{i,j}= -\sum_{z \in \{0,1\}} \int_{\cA } p(z|a,\theta)\frac{\partial^2 \log p(z|a,\theta)}{\partial\theta_i \partial\theta_j} p(a)da
        \end{equation}
        for $i,j\in\{1,2\}$, with $p(z|a,\theta)=\ell_1(z|a,\theta)$ being the likelihood expressed in Eq.~(\ref{eq:likelihood}), i.e.,
            \begin{multline}
                \log p(z|a,\theta) = z\log\Phi\left(\frac{\log a-\log\alpha}{\beta}\right)\\ + (1-z)\log\left(1-\Phi\left(\frac{\log a-\log\alpha}{\beta}\right)\right).
            \end{multline}
            
Denoting $\gamma=\gamma(a)=\beta^{-1}\log({a}/{\alpha})$, the first-order partial derivatives of $\log p(z|a,\theta)$ with respect to $\theta$ are:
        \begin{align}
            \frac{\partial}{\partial\alpha}\log p(z|a,\theta) =& -\frac{1}{\alpha\beta}z\frac{\Phi'(\gamma)}{\Phi(\gamma)} + \frac{1}{\alpha\beta}(1-z)\frac{\Phi'(\gamma)}{1-\Phi(\gamma)} , \\
            \nonumber
            \frac{\partial}{\partial\beta}\log p(z|a,\theta) =& -\frac{\log\frac{a}{\alpha}}{\beta^2}z\frac{\Phi'(\gamma)}{\Phi(\gamma)}\\ &+ \frac{\log\frac{a}{\alpha}}{\beta^2}(1-z)\frac{\Phi'(\gamma)}{1-\Phi(\gamma)} ,
        \end{align}
        and the second-order partial derivatives are:
        \begin{align}
        \nonumber
            &\frac{\partial^2}{\partial \alpha\partial\beta}\log p(z|a,\theta) \\ &=-\frac{1}{\beta}\frac{\partial}{\partial\alpha}p(z|a,\theta) 
            + \frac{\log\frac{a}{\alpha}}{\alpha\beta^3}z\frac{\Phi''(\gamma)\Phi(\gamma)-\Phi'(\gamma)^2}{\Phi(\gamma)^2}\label{eq:partialalphbet}\\
            &\hspace*{1em}
                - \frac{\log\frac{a}{\alpha}}{\alpha\beta^3}(1-z)\frac{\Phi''(\gamma)(1-\Phi(\gamma))+\Phi'(\gamma)^2}{(1-\Phi(\gamma))^2} ,\nonumber
        \end{align}
        \begin{align}
        \nonumber
            &\frac{\partial^2}{\partial\alpha^2}\log p(z|a,\theta) \\ &=  -\frac{1}{\alpha}\frac{\partial}{\partial\alpha}\log p(z|a,\theta) 
            + \frac{1}{\alpha^2\beta^2}z\frac{\Phi''(\gamma)\Phi(\gamma)-\Phi'(\gamma)^2}{\Phi(\gamma)^2}\label{eq:partial2alph}\\
            &\hspace*{1em}
                - \frac{1}{\alpha^2\beta^2}(1-z)\frac{\Phi''(\gamma)(1-\Phi(\gamma))+\Phi'(\gamma)^2}{(1-\Phi(\gamma))^2} , \nonumber
        \end{align}  
        and
        \begin{align}
        \nonumber
            &\frac{\partial^2}{\partial\beta^2}\log p(z|a,\theta)\\ &= 
            -\frac{2}{\beta}\frac{\partial}{\partial\beta}\log p(z|a,\theta) 
            + \frac{\log^2\frac{a}{\alpha}}{\beta^4}z\frac{\Phi''(\gamma)\Phi(\gamma)-\Phi'(\gamma)^2}{\Phi(\gamma)^2} \label{eq:partial2bet}\\
            &
            - \frac{\log^2\frac{a}{\alpha}}{\beta^4}(1-z)\frac{\Phi''(\gamma)(1-\Phi(\gamma))+\Phi'(\gamma)^2}{(1-\Phi(\gamma))^2}.\nonumber
        \end{align}

    The expressions in Eqs.~(\ref{eq:partialalphbet}), (\ref{eq:partial2alph}), and (\ref{eq:partial2bet}) of the second-order partial derivatives of $p(z|a,\theta)$ need to be integrated over $\cZ$ and $\cA$. Summing over the discrete variable $z$ first replaces $z$ by $\Phi(\gamma)$ in the equations.
    Finally, if we denote $A_{01}$, $A_{02}$, $A_{11}$, $A_{12}$, $A_{21}$, $A_{22}$ by
        \begin{equation}\label{eq:Aij}
        \begin{aligned}
            A_{0u} &= \int_\cA\frac{\Phi'(\gamma(a))^2}{\Phi((-1)^{u+1}\gamma(a))}p(a)da,\\
            A_{1u} &= \int_\cA\log\frac{a}{\alpha}\frac{\Phi'(\gamma(a))^2}{\Phi((-1)^{u+1}\gamma(a))}p(a)da,\\
            A_{2u} &= \int_\cA\log^2\frac{a}{\alpha}\frac{\Phi'(\gamma(a))^2}{\Phi((-1)^{u+1}\gamma(a))}p(a)da,
        \end{aligned}
        \end{equation}
    for $u\in\{1,2\}$,
    then the information matrix $\cI(\theta)$ takes on the following form:
        \begin{equation}
        \label{eq:infmat}
            \cI(\theta)=\begin{pmatrix}
            \frac{1}{\alpha^2\beta^2}(A_{01} + A_{02}) & \frac{1}{\alpha\beta^3}(A_{11}+A_{12}) \\
            \frac{1}{\alpha\beta^3}(A_{11}+A_{12}) & \frac{1}{\beta^4}(A_{21}+A_{22})
        \end{pmatrix}.
        \end{equation}
    The integrals in Eq.~(\ref{eq:Aij}) are computed using Simpson's rule on a regular grid. The distribution $p(a)$ is approximated by kernel density estimation based on seismic signals.
    Finally, the Jeffreys prior is obtained by taking the square root of the determinant of the matrix defined in Eq.~(\ref{eq:infmat}).

        \subsection{Practical implementation}\label{sec:practseism}
        
        Section~\ref{sec:objprior} showed that knowing the probability distribution of the IM is required in order to calculate the Fisher information matrix. Without compromising on the general applicability of the methodology, let us consider in the following section the PGA as the IM. {Still, it is important to emphasize that this choice is purely illustrative and bears no consequence for the proposed methodology.} In this study, we used $10^5$ artificial seismic signals generated using the stochastic generator defined in \cite{Rezaeian2010} and implemented in~\cite{Sainct2020} from 97 real accelerograms selected in the European Strong Motion Database for $5.5 \leq \text{M} \leq 6.5$ and $\text{R} < 20$~km. The generation of seismic ground motions isn’t a necessity in the Bayesian framework (especially if a sufficient number of real signals are available), but it allows comparisons with the Monte-Carlo (MC) reference method for simulation-based approaches, as well as comparative studies of performance. Note that the artificial signals have the same PGA distribution as the real ones, as shown in Figure~\ref{fig:IM}.

        \begin{figure}
            \centering
            \includegraphics[width=180pt]{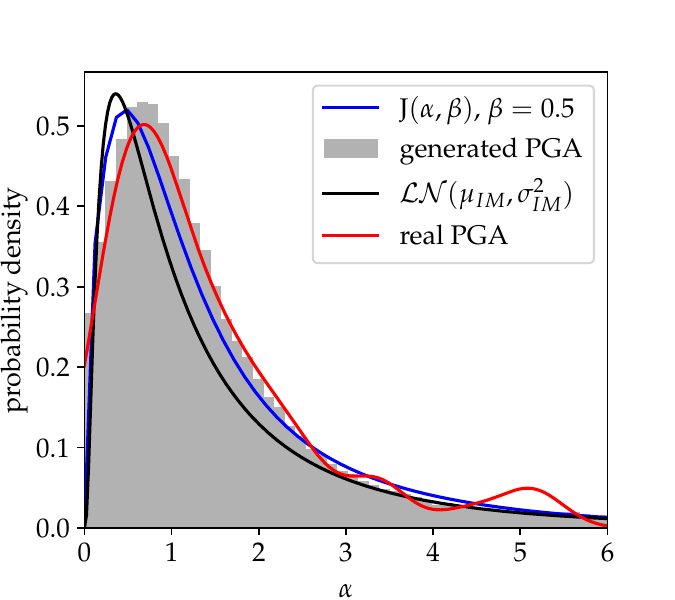} %
            \caption{Comparison of a sectional view of the Jeffreys prior w.r.t. $\alpha$ (blue line) with some PGA distributions: approximated distribution of real accelerograms via Gaussian kernels in red, histogram of the generated signals in gray, and log-normal fit of the distribution in black.
            }
            \label{fig:IM}
        \end{figure}

        In practice, the use of Markov Chain Monte-Carlo (MCMC) methods (see Section~\ref{sec:BayesFram}) to sample the \emph{a posteriori} distribution means that the prior must be evaluated (up to a multiplicative constant) many times during the calculations. Considering the computational complexity stemming from the integrals that need to be calculated, it was decided to perform the evaluations of the prior on an experimental design based on a fine-mesh grid of $\Theta$ (here $[0,+\infty)^2$) and to build an interpolated approximation of the Jeffreys prior matching this design. This strategy is more suitable for our numerical applications and very tractable because $\Theta$ is a two-dimensional domain. Figure~\ref{fig:jeff_prior} shows a plot of the Jeffreys prior. To be precise, $500\times500$ prior values were computed for $\alpha\in[10^{-5},10]$ and $\beta\in[10^{-3},2]$ and then processed in order to obtain a linear interpolation. 
     
        \begin{figure} %
            \centering
            \includegraphics[width=220pt]{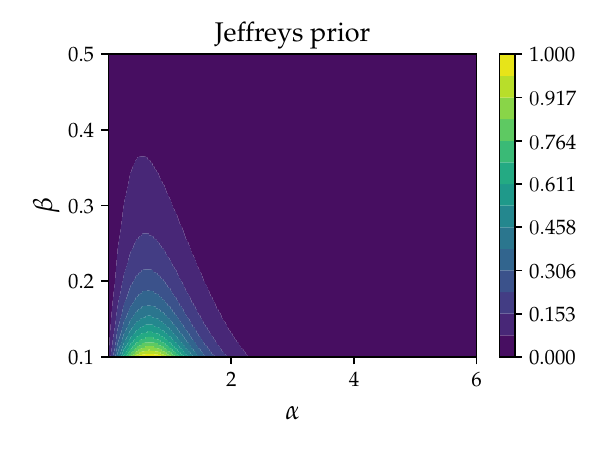}%
            \caption{The Jeffreys prior calculated from PGA and plotted in the subdomain $[0,6]\times [0.1,0.5]$.}
            \label{fig:jeff_prior}
        \end{figure}
        
        \subsection{Discussion}\label{sec:JeffDiscussion}

        The computational complexity of the Jeffreys prior is not in itself a major drawback. Since it depends exclusively on the distribution of the IM, the initial cost of the Jeffreys prior’s complex calculation would quickly be recovered in installations on the scale of a nuclear power plant, where it is necessary to determine the fragility curves of a large number of Structures and Components (SCs). Compared to methodologies that aim to define a prior based on mechanical calculations which are, by definition, specific to SCs, the generic character of the Jeffreys prior is a clear advantage. This will be explored in the applications section of this paper (Section~\ref{sec:application}). Moreover, the Jeffreys prior is completely defined and does not depend on any additional subjective choices. 
        
        The Jeffreys prior is known to be improper in numerous common cases (i.e., it cannot be normalized as a probability). This is relevant to our study, as evidenced in \ref{app:asymptotics}, where a calculation of the asymptotic behavior for different limits of $\theta = (\alpha,\beta)$ shows that the Jeffreys prior is indeed improper in our case. However, this characteristic is not an issue since our study focuses on the posterior, which is itself proper, as proven in \ref{app:asymptotics}. This is a critical issue, since MCMC algorithms would not make any sense if the posterior were improper. These asymptotic expansions also provide complementary and essential insight into the Jeffreys prior. They make it possible to understand that its behavior in $\alpha$ is similar to that of a log-normal distribution having the same median as that of the IM (i.e., here 1.1~m/s$^2$) with a variance calculated as the sum of the variance of the IM and of a term that depends on $\beta$. Figure~\ref{fig:IM} clearly illustrates this result.

\section{Estimation tools, competing approaches and performance evaluation metrics}\label{sec:tools}

    In this section, we will first present the Bayesian estimation tools and the MC reference method used to evaluate the relevance of the log-normal model when the amount of data allows it. We will then present two competing approaches, implemented in order to evaluate the performance of the Jeffreys prior in practical cases. On the one hand, we will apply the MLE method, widely used in literature, coupled with a bootstrap technique. On the other hand, we will apply a Bayesian technique implemented with the prior introduced by \citeauthor{Straub2008} \cite{Straub2008}. For a fair comparison, this study proposes to calibrate the latter according to the results of Figure~\ref{fig:IM}, which illustrates that in $\alpha$ the distribution is similar to the PGA distribution of the artificial and real signals. It would indeed be easy to calibrate it in such a way so as to skew comparisons, for instance by considering too large a variance. Finally, we will define performance evaluation metrics.
    
    \subsection{Fragility curves estimations via Monte-Carlo}
    \label{sec:reference}

        Here, let us assume the availability of a validation data set $(\mbf a^{\mathrm{MC}},$ $\mbf z^{\mathrm{MC}})$ $ =$ $ ( (a_i^{\mathrm{MC}})_{i=1}^{N^{\mathrm{MC}}} $, $(z_i^{\mathrm{MC}})_{i=1}^{N^{\mathrm{MC}}})$. This section describes how a fragility curve can be obtained from such a large data set by non-parametric estimations that can serve as a reference. This way, our estimations (based on a small data set and parametric estimations) can be compared with this reference.     
        Good candidates for estimating this reference are MC estimators, which estimate the expected number of failures locally w.r.t. the IM.

        First, we need to divide the IM values into sub-intervals and estimate the probability of failure for each. Sub-intervals of regular size should be avoided because the observed IMs are not uniformly distributed. We will therefore consider clusters of IMs, defined through the K-means, as suggested by \citeauthor{TREVLOPOULOS2019} \cite{TREVLOPOULOS2019}. 
        Given $N_c$ such clusters $(K_j)_{j=1}^{N_c}$, the MC fragility curve estimated at the centroids $(c_j)_{j=1}^{N_c}$ is expressed as
            \begin{equation}\label{eq:refMC}
                P_f^{\mathrm{MC}}(c_j) = \frac{1}{n_j}\sum_{i,\,a_i^{\mathrm{MC}}\in K_j}z_i^{\mathrm{MC}}  , 
            \end{equation}
        where $n_j$ is the sample size of cluster $K_j$. An asymptotic confidence interval for this estimator can also be derived from its Gaussian approximation. It is accepted that an MC-based fragility curve can be considered a reference curve because it is not based on assumptions.

    \subsection{Fragility curves estimations in the Bayesian framework} \label{sec:BayesFram}

        The most relevant method in order to benefit from the Bayesian theory introduced in Section~\ref{sec:objprior} and the reference prior construction presented in Section~\ref{sec:construction} consists in deriving the posterior defined in Eq.~(\ref{eq:posterior}). It then becomes possible to generate, according to that distribution, samples of $\theta$ conditioned on the observed data. These \emph{a posteriori} generations of $\theta$ can be obtained using MCMC methods. In this study, we have implemented an adaptive Metropolis-Hastings (M-H) algorithm with a Gaussian transition kernel and a covariance adaptation process \cite{Haario2001}. Such an algorithm allows sampling from a probability density known up to a multiplicative constant. In this context, the \emph{a posteriori} samples of $\theta$ can be used to define credibility intervals for the log-normal estimations of the fragility curves.
    
        \subsection{Multiple MLE by bootstrapping}\label{sec:bootstrap}

            The best known parameter estimation method is the MLE, defined as the maximal argument of the likelihood derived from the observed data:
            \begin{equation}\label{eq:MLE}
                \hat\theta^{\mathrm{MLE}}_k = \argmax_{\theta\in\Theta} \ell_k(\mbf z|\mbf a, \theta).
            \end{equation}

            A common method for obtaining a wide range of $\theta$ estimations consists in calculating multiple MLEs by bootstrapping. Denoting the data set size by $k$, bootstrapping consists in doing $L$ independent draws with the replacement of $k$ items within the data set. These draws lead to $L$ different likelihoods from the $k$ initial observations, and so to $L$ values of the estimator, which can then be averaged. This is a very common approach for fragility curves (e.g., \cite{Shinozuka2000, Lallemant2015, Gehl2015, Baker2015, WangF2020}). The convergence of the MLE and the relevance of this method are detailed in \cite{VanDerVaart1992}. However, the relevance of the bootstrap method is often limited by the irregularity of its results for small values of $k$ (e.g., \cite{Zentner2017}). In this context, the $L$ values of $\theta$ are used to define confidence intervals for the log-normal estimations of the fragility curves.
        
        \subsection{{Example of a prior found in literature for log-normal seismic fragility curves}} %
        \label{sec:posterioriSimul}
        {For comparison purposes, we selected the prior} suggested by \citeauthor{Straub2008}---called the SK prior---which is defined as the product of a normal distribution for $\ln(\alpha)$ and the improper distribution $1/\beta$ for $\beta$, namely:
                \begin{equation}\label{eq:Straubprior}
                    \pi_{SK}(\theta)\propto\frac{1}{\alpha\beta} \exp\Big( -\frac{(\log\alpha-\mu)^2}{2\sigma^2}\Big).
                \end{equation}
        In \cite{Straub2008}, the parameters $\mu$ and $\sigma$ of the log-normal distribution are chosen to generate a non-informative prior. 
        As specified in the introduction to Section~\ref{sec:tools}, for a fair comparison with the approach proposed in this paper, we decided to pick $\mu$ and $\sigma$ as equal to the mean and the standard deviation of the logarithm of the IM. This choice is consistent with the fact that the Jeffreys prior is similar to a log-normal distribution with these parameters (see Figure~\ref{fig:IM}). The prior $\pi_{SK}(\theta)$ is plotted in Figure~\ref{fig:Straubprior}.

            \begin{figure}
                \centering
                \includegraphics[width=210pt]{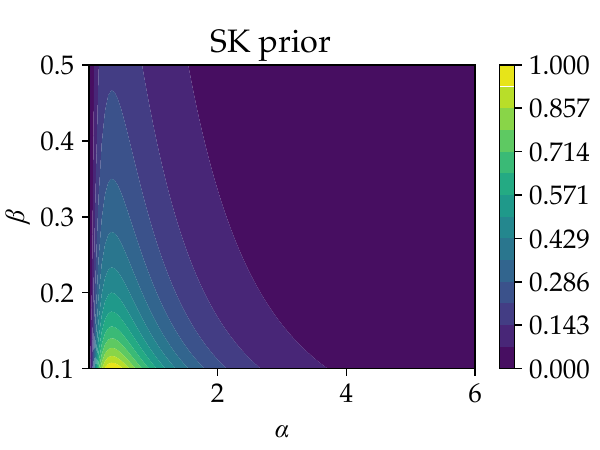} %
                \caption{Prior suggested by \citeauthor{Straub2008} \cite{Straub2008} and plotted in the subdomain $[0,6]\times [0.1,0.5]$. It is expressed in Eq.~(\ref{eq:Straubprior}) and scaled to a log-normal estimation of the PGA's distribution.}
                \label{fig:Straubprior}
            \end{figure}

        An analysis of the posterior obtained from the SK prior is given in \ref{app:asymptotics}. It shows that the posterior is improper, which jeopardizes the validity of any \emph{a posteriori} estimation using MCMC methods. This could however be mitigated by truncating w.r.t. $\beta$. This issue persists in the authors' original framework, which is slightly different from ours. This was confirmed in \ref{app:SKreview}.
            
    \subsection{Performance evaluation metrics}\label{sec:metrics}
    
        In order to obtain a clear view of the performance of the proposed approach, we considered two quantitative metrics that can be calculated for each of the methods described in the previous subsections. Considering the sample $(\mbf a, \mbf z) $, we denote by $a \mapsto P_f^{|\mbf a,\mbf z}(a)$ the random process defined as the fragility curve conditional to the sample (the probability distribution of $P_f^{|\mbf a,\mbf z}(a)$ is inherited from the \emph{a posteriori} distribution of $\theta$). For each value $a$ the $r$-quantile of the random variable $P_f^{|\mbf a,\mbf z}(a)$ is denoted by $q_r^{|\mbf a,\mbf z}(a)$. We can then define:
    \begin{itemize}
            \item The conditional quadratic error:
                \begin{align}\label{eq:quaderror}
                    \cE^{Q|\mbf a,\mbf z} &= \EE\big[\| P_{f}^{|\mbf a,\mbf z} - P_f^{\mathrm{MLE}} \|_{L^2}^2|\mbf a,\mbf z \big]\\ &= \int_{0}^{A_{\rm max}} \EE\big[ (P_{f}^{|\mbf a,\mbf z}(a) - P_f^{\mathrm{MLE}} (a))^2 |\mbf a,\mbf z \big] da.\nonumber
                \end{align}
        $P_f^{\mathrm{MLE}}$ stands for the log-normal estimation of the fragility curve obtained by the MLE (see Section~\ref{sec:reference}) taking into account all the data available in the case study. We further checked that this estimation was close to the reference curve obtained by MC whenever possible (see Section~\ref{sec:application}).
                            
            \item  The conditional width of the $1-r$ credibility zone for the fragility curve:
                \begin{align}\label{eq:scaleerror}
                    \cS^{r|\mbf a,\mbf z} &= \|q_{1-{r/2}}^{|\mbf a,\mbf z} - q_{{r/2}}^{|\mbf a,\mbf z}\|_{L^2}^2\\  &= \int_{0}^{A_{\rm max}} ( q_{1-{r/2}}^{|\mbf a,\mbf z} (a) - q_{{r/2}}^{|\mbf a,\mbf z}(a))^2 da  .\nonumber
                \end{align}
        \end{itemize}

    To estimate such variables, we simulated a set of $L$ %
    fragility curves $( P_f^{\theta_i|\mbf a,\mbf z})_{i=1}^L$ where $(\theta_i)_{i=1}^L$ is a sample of the \emph{a posteriori} distribution of the model parameters obtained by MCMC. The empirical quantiles $\hat q_r^{\theta_{i=1}^L|\mbf a,\mbf z}(a)$  of $( P_f^{\theta_i|\mbf a,\mbf z}(a))_{i=1}^L$  are approximations of the quantiles $q_r^{|\mbf a,\mbf z}(a)$ of the random variable $P_f^{|\mbf a,\mbf z}(a)$. We derive:
        \begin{itemize}
            \item The approximated conditional quadratic error:
                \begin{equation}\label{eq:quaderrorapprox}
                    \hat\cE_{L}^{Q|\mbf a,\mbf z} =\frac{1}{L}\sum_{i=1}^L\| P_{f}^{\theta_i|\mbf a,\mbf z} - P_f^{\mathrm{MLE}} \|_{L^2}^2.
                \end{equation}
            \item The approximated conditional width of the $1-r$ credibility zone for the fragility curve:
                \begin{equation}\label{eq:scaleerrorapprox}
                    \hat\cS_L^{r|\mbf a,\mbf z} = \|\hat q_{1-{r/2}}^{\theta_{i=1}^L|\mbf a,\mbf z} - \hat q_{{r/2}}^{\theta_{i=1}^L|\mbf a,\mbf z} \|_{L^2}^2.
                \end{equation}
        \end{itemize}
        The $L^2$ norms are integrals over $a\in [0,A_{\rm max}]$ which are approximated numerically using Simpson's interpolation on sub-intervals of regular size  $0=A_0<\dots<A_p=A_{\rm max}$. In the following examples, we shall use $A_0=0$, $A_{\rm max}=12\,\mathrm{m/s^2}$, and $p=200$.

        For the MLE with bootstrapping, we can define a conditional quadratic error similarly to Eq.~(\ref{eq:quaderrorapprox}) and a conditional width of the $1-r$ confidence interval similarly to Eq.~(\ref{eq:scaleerrorapprox}) using a bootstrapped sample $(\theta_i)_{i=1}^L$.

\section{Numerical applications}\label{sec:application}

In this section, we will examine two case studies.
These leverage the many simulation data sets available that have been previously computed for validation purposes. They will be used in the derivation of a reference fragility curve (as suggested in Section~\ref{sec:reference}), and allow us to validate the corresponding log-normal models. The first case, described in Section~\ref{sec:onl}, deals with a nonlinear oscillator,
for which $N_s=10^5$ nonlinear simulations have been implemented for validation purposes. The second case study, described in Section~\ref{sec:ASG}, deals with a piping system which is part of the secondary cooling system of a French Pressurized Water Reactor. Due to the high computational cost, only $N_s=10^4$ simulations have been performed for this case. In both cases, estimations are performed using different testing data sets of a size $k$ chosen as negligible compared to $N_s$. These testing data sets are taken from the set of available nonlinear dynamical simulation results. {A third case study is presented as supplementary material in order to showcase how our method could be applied to practical experiments.}

\subsection{Case study 1: nonlinear oscillator \label{sec:onl}}

    This first case study---depicted in Figure~\ref{fig:KH}---relates to a single-degree-of-freedom elastoplastic oscillator with kinematic hardening. This mechanical system illustrates the essential features which can be found in the nonlinear responses of some real-world structures under seismic excitation and has, for this reason, already been used in several studies~\cite{TREVLOPOULOS2019, Sainct2020, Gauchy2021}. The motion of a unit mass $m$ can be described by the equation:
    \begin{equation}
    \ddot{y}(t) + 2 \zeta \omega_{\text{L}}\dot{y}(t) + f(t) = -s(t) \ ,
    \end{equation}
    with $s(t)$ a seismic signal, $\dot{y}(t)$ and $\ddot{y}(t)$ respectively the relative velocity and acceleration of the mass, $\zeta$ the damping ratio, and $\omega_{\text{L}}$ the circular frequency. The relevant EDP is the absolute maximum value of the mass’ displacement, i.e., $\max_{t\in [0, T]}|y(t)|$, where $T$ is the duration of the seismic excitation. The failure criterion $C$ is chosen to be the $90\%$-level quantile of the maximum displacement calculated with $10^5$ artificial signals, i.e., $C = 8.0 \; 10^{-3}$~m.
    Figure \ref{fig:ref-osci} compares the MC-based reference fragility curve $P_f^{\mathrm{MC}}$ (Eq.~(\ref{eq:refMC})) with its log-normal estimation $P_f^{\mathrm{MLE}}$, both estimated using the results of the $10^5$ simulations. In this case, the log-normal fragility curve is a good approximation of the reference curve.

    \begin{figure} %
    \centering
    \includegraphics[width=5cm]{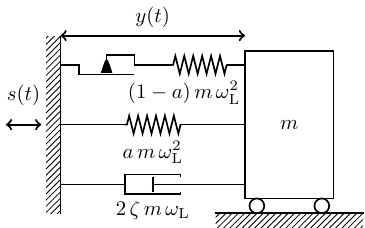}
    \caption{Elastoplastic oscillator with kinematic hardening, with parameters $f_{\text{L}} = 5$ Hz and $\zeta = 2\%$. The yield limit is $Y = 5.10^{-3}$~m, and the post-yield stiffness is $20\%$ of the elastic stiffness, i.e., $a = 0.2$.}
    \label{fig:KH}
    \end{figure}

    \begin{figure}
        \centering
        \includegraphics[width=180pt]{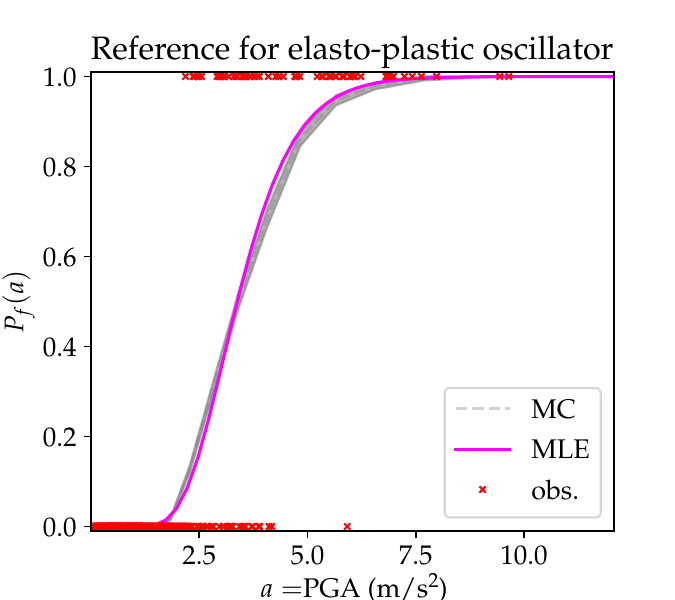} %
        \caption{Reference fragility curve $P_f^{\mathrm{MC}}$ (see Section~\ref{sec:reference}) compared with $P_f^{\mathrm{MLE}}$ (see Section~\ref{sec:metrics}) for the elastoplastic oscillator (case study 1). Both curves were computed using the full data set generated ($10^5$ items). The red crosses represent the observations; 0 means no failure, while 1 means failure.}
        \label{fig:ref-osci}
    \end{figure}

    The fragility curve estimations are shown in Figure~\ref{fig:curvesONL}. They are obtained from $L=5000$ samples of $\theta$ generated with the implemented statistical methods (see Section~\ref{sec:tools}), which are based on two samples of nonlinear dynamical simulations of sizes $k=20$ and $k=30$. Although the nature of the two intervals compared is different---credibility interval for the Bayesian framework and confidence interval for the MLE---, these results clearly illustrate the advantage of the Bayesian framework over the MLE for small samples. With the MLE, irregularities characterized by null estimates of $\beta$ appear, resulting in "vertical" confidence intervals. In \ref{app:asymptotics}, we established that the likelihood is easily maximized for $\beta=0$ when samples are partitioned into two disjunct subsets when classified according to IM values: one subset for which there is no failure and one for which there is failure. Moreover, when few failures are observed in the initial sample, the bootstrap technique can lead to the generation of a large number of samples that maximize the likelihood at $\beta=0$. This is better evidenced by an examination of the raw values of $\theta$ generated in Figure \ref{fig:scatterONL}. The degenerate $\beta$ values resulting from the MLE appear clearly but, although it should theoretically also be affected, the Bayesian framework shows no evidence of a similar phenomenon for this type of samples.

            \begin{figure*}
                \centering%
                \includegraphics[width=130pt]{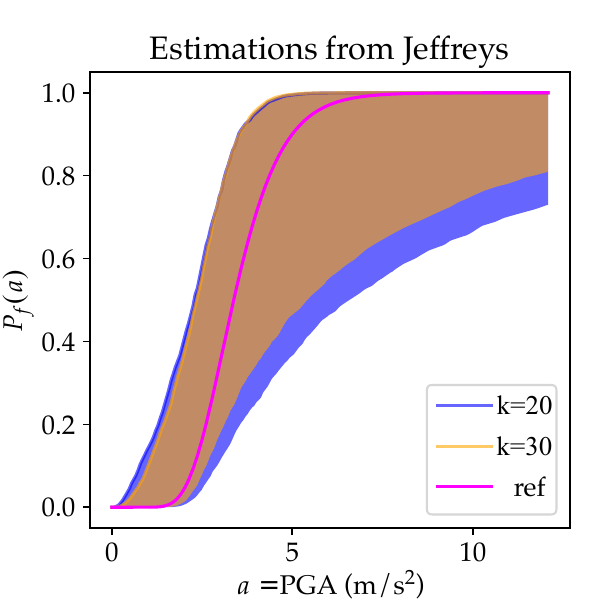}%
                \includegraphics[width=130pt]{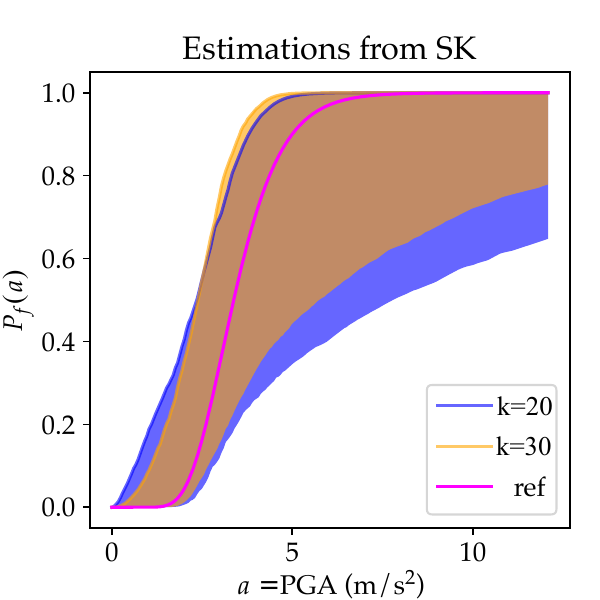}%
                \includegraphics[width=130pt]{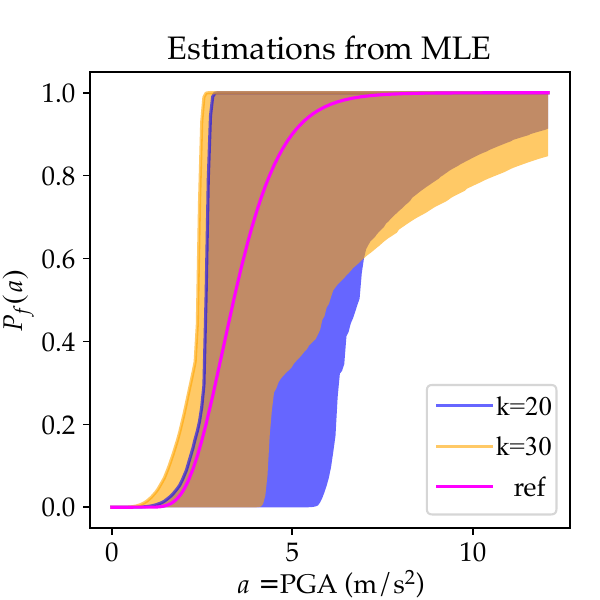}%
                \caption{$95\%$ credibility (for Bayesian estimation) or confidence (for the MLE) intervals of fragility curve estimations for the elastoplastic oscillator, obtained from a total of $L=5000$ estimations of $\theta$ using the {statistical methods introduced in Section~\ref{sec:reference}}: (from left to right) Bayesian estimation using the Jeffreys prior, Bayesian estimation using the SK prior, and MLE with bootstrapping. For each of these, we considered two data samples {of nonlinear dynamical simulations} of two different sizes ($k=20$ in blue, $k=30$ in orange). $P_f^{\mathrm{MLE}}$ {(see Section \ref{sec:metrics}) is plotted in magenta}.}
                \label{fig:curvesONL}
            \end{figure*}

            \begin{figure*}
                \centering%
                \includegraphics[width=130pt]{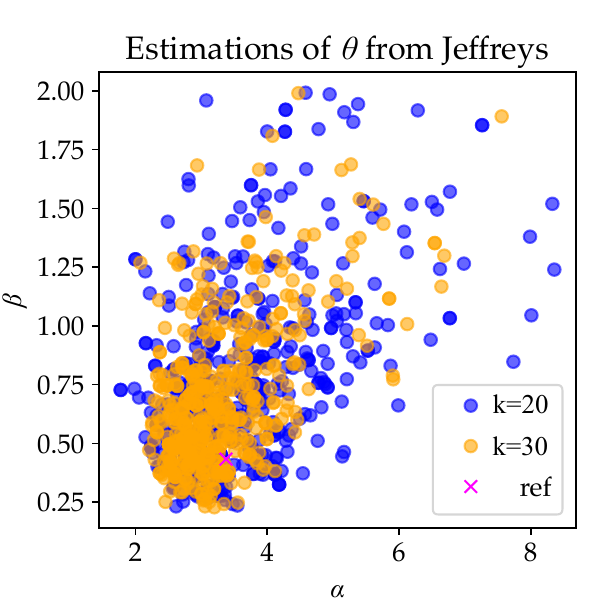}%
                \includegraphics[width=130pt]{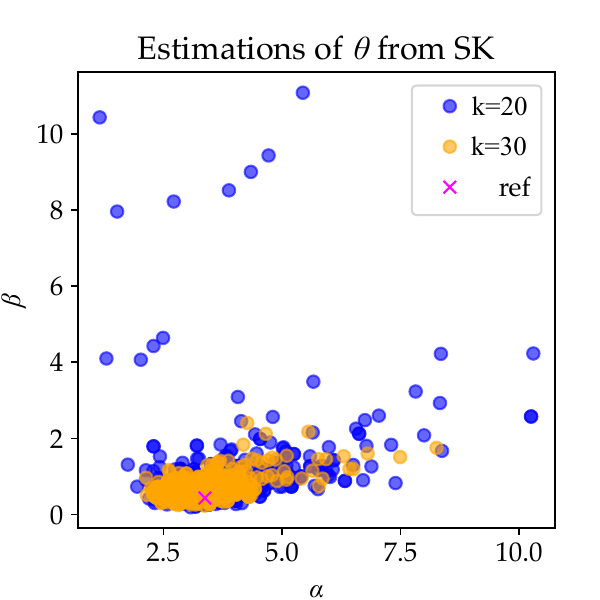}%
                \includegraphics[width=130pt]{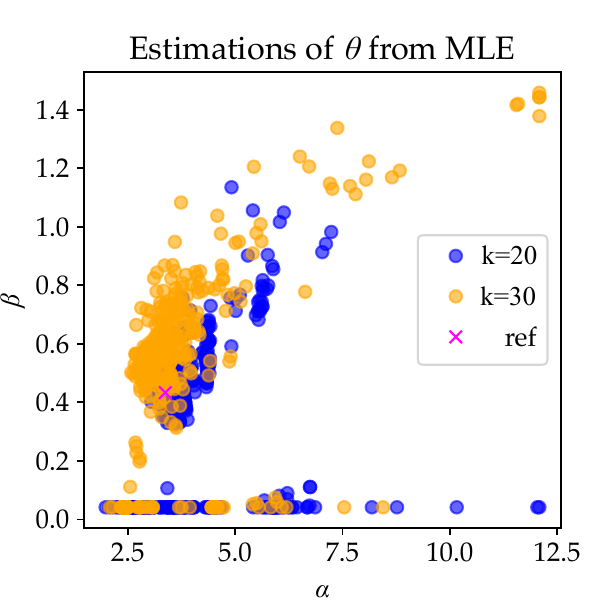}%
                 \caption{Scatter plots of the generated $\theta$ for the estimation of the fragility curves presented in Figure~\ref{fig:curvesONL} for the elastoplastic oscillator. For all three statistical methods, we plotted $500$ points out of the $L=5000$ $\theta=(\alpha,\beta)$ {estimated} with two data sets of nonlinear dynamical simulations (of size $k=20$ in blue and $k=30$ in orange). The magenta crosses represent $\theta^{\mathrm{MLE}}$, used for the computation of $P_f^{\mathrm{MLE}}$ (see Section~\ref{sec:metrics}). This figure reveals both the outliers generated from the SK prior (center) and the irregularities characterized by null estimates of $\beta$ for the coupled MLE and bootstrap approach (right).}
                \label{fig:scatterONL}
            \end{figure*}

        Since the SK prior is calibrated to look like the Jeffreys prior, Figure~\ref{fig:curvesONL} shows a strong similarity between the Bayesian estimations of the fragility curves obtained with these two priors. However, Figure~\ref{fig:scatterONL} (middle) shows that many outliers are obtained with the SK prior. These values explain why the credibility intervals of the fragility curves are larger with the SK prior when $k = 20$. This observation is supported theoretically by the calculation provided in \ref{app:asymptotics}. There is actually a better convergence of the Jeffreys prior toward $0$ when $\beta\to\infty$. This superior asymptotic behavior obviously results in posteriors that happen to assign a lower probability to outlier points (a phenomenon particularly noticeable when the data sample is small) as well as to the weight of the likelihood within the posterior.
        
        For a better understanding of this phenomenon, we calculated the quantitative metrics defined in Section~\ref{sec:metrics}. For any $k$ ranging from $15$ to $100$, we conducted $m=200$ different draws of observations $(\mbf a^{(j)},\mbf z^{(j)})_{j=1}^m$ in order to derive the metrics $\hat\cE_{L,R}^{Q|\mbf a^{(j)},\mbf z^{(j)}}$, $\hat\cS_{L,R}^{r|\mbf a^{(j)},\mbf z^{(j)}}$, $j\in\{1,\dots,m\}$, $R\in\{$`MLE', `SK', `Jeffreys'$\}$, $L=5000$, $1-r=95\%$. 
        The corresponding means and $95\%$-confidence intervals are plotted in Figure~\ref{fig:errors}. Firstly, these diagrams demonstrate the benefits of the Bayesian framework compared to the MLE approach for small observation sets. Secondly, the compared performance of the Jeffreys and SK posteriors is highlighted by the confidence interval endpoints of the quadratic error and the credibility interval. Specifically, the latter highlights the effect of the superior asymptotic behavior of the Jeffreys prior along the width of the credibility interval. It shows variations similar to but smaller than the SK prior, thus highlighting its capacity to generate fewer outliers for the pair $(\alpha, \beta)$, as expected.

        \begin{figure*}
                \centering%
                \includegraphics[width=154pt]{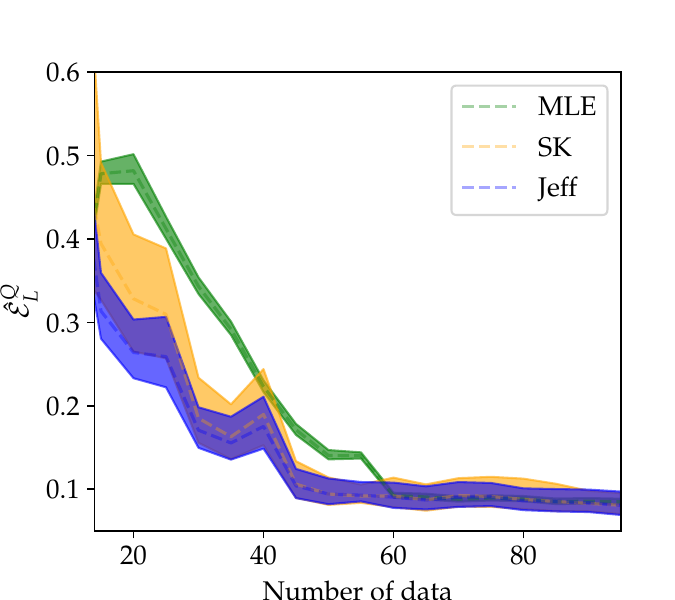} %
                \includegraphics[width=154pt]{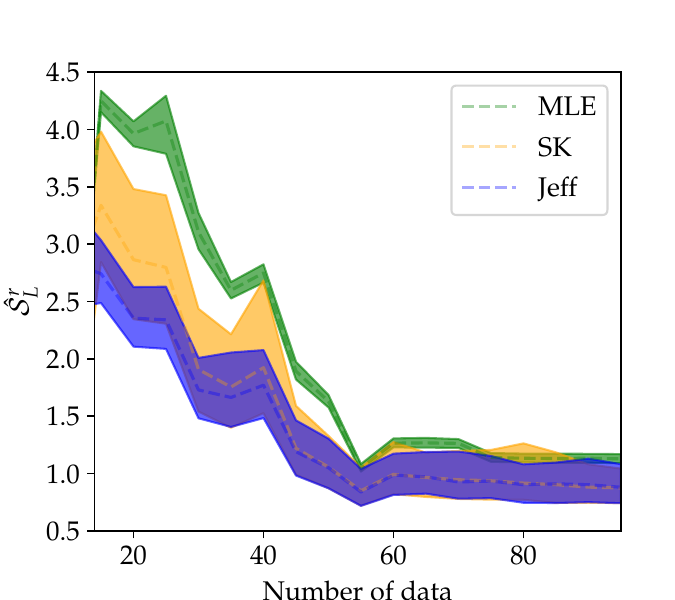}%
                \caption{%
                Performance evaluation metrics (see Section~\ref{sec:metrics}) for the elastoplastic oscillator computed by replications from independent draws in the full data set {of nonlinear dynamical simulations} and for sample sizes ranging from $k=15$ to $100$. Left: the dashed lines plot the quadratic errors as a function of the number of observations, and the shaded areas show their confidence intervals. Right: the dashed lines plot the widths of the credibility (for the Bayesian estimation) or confidence intervals
                (for the MLE), and the shaded areas show their confidence intervals.}
                \label{fig:errors}
            \end{figure*}

\subsection{Case study 2 : piping system \label{sec:ASG}}

This second case study deals with a piping system which is part of the secondary cooling system of a French Pressurized Water Reactor. This piping system was studied experimentally and numerically within the framework of the ASG program \cite{TOUBOUL1999}. Figure~\ref{fig:ASG} (left) shows a view of the piping mock-up, mounted on the Azalee shaking table of the EMSI laboratory of CEA/Saclay, whereas the finite element model (FEM)---based on beam elements---is shown in Figure~\ref{fig:ASG} (right). The FEM was implemented using CAST3M~\cite{CAST3M}, a homemade FE code, and validated through an experimental campaign.

	\begin{figure*} %
		\centering		
		\includegraphics[width=6cm]{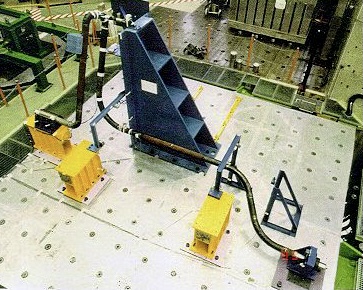}
		\hspace{0.5cm}
		\includegraphics[trim= 1cm  3.8cm 12cm 1.5cm, clip,width=5cm]{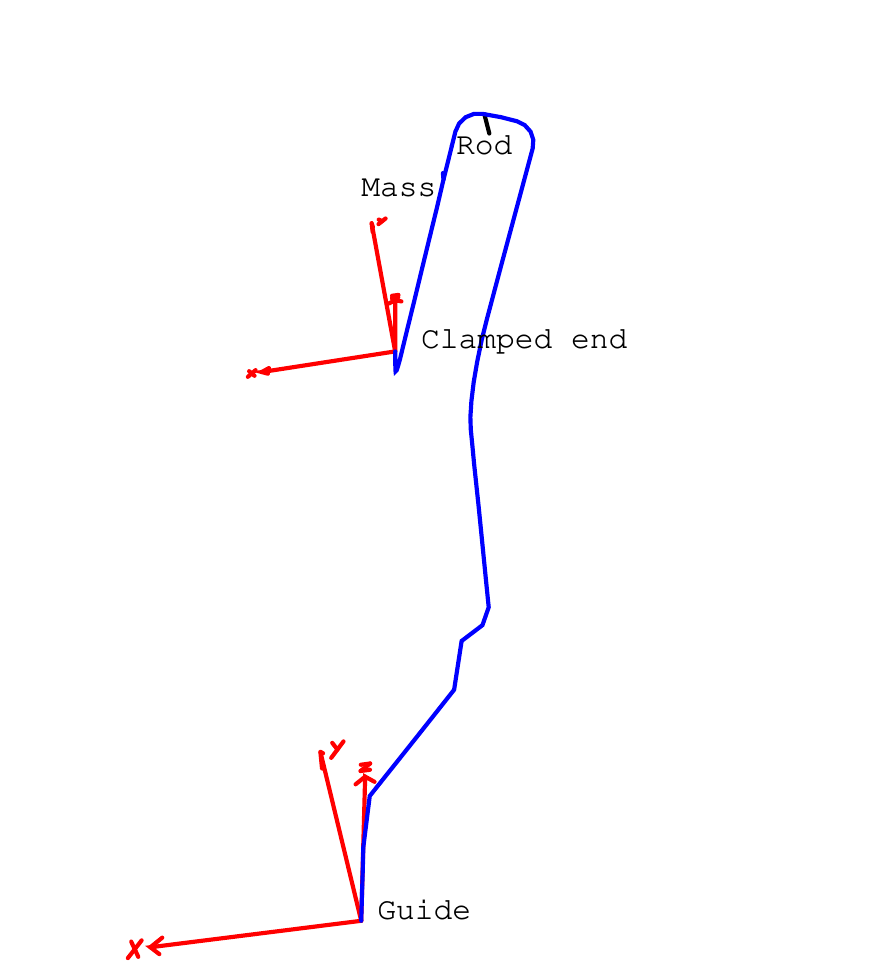}
		\caption{Overview of the piping system on the Azalee shaking table (left) and mock-up FEM (right).}
		\label{fig:ASG}
	\end{figure*}

    The mock-up is a pipe filled with unpressurized water. The pipe has an external diameter of 114.3~mm and a thickness of 8.56~mm, with an elbow characteristic parameter of 0.47, and is made of TU42C carbon steel. There are three elbows along the pipe as well as a 120~kg mass standing in for a valve. This mass accounts for more than 30\% of the total mass of the mock-up. One end of the mock-up is clamped while the other is supported by a guide in order to prevent displacements along the X and Y axis. In addition, a rod is placed on top of the mock-up to limit the mass displacements along the Z axis (see Figure~\ref{fig:ASG} (right)). During the tests, the excitation is only imposed along the X axis. For this study, the artificial signals are filtered by a fictional, 2\%-damped linear single-mode building at 5~Hz, 5~Hz being the first eigenfrequency of the 1\%-damped piping system. The failure criterion is considered to be an excessive out-of-plane rotation of the elbow located near the clamped end of the mock-up, as recommended in \cite{TOUBOUL2006}. The critical rotation considered is $4.1^\circ$. This value is the 90\%-level quantile of a sample of nonlinear numerical simulations of size $10^4$. 
    
    Figure \ref{fig:ref-ASG} shows the comparison between the reference, MC-based fragility curve $P_f^{\mathrm{MC}}$ (Eq.~(\ref{eq:refMC})) and its log-normal estimation $P_f^{\mathrm{MLE}}$, both estimated using the results of $10^4$ simulations. Here, the log-normal fragility curve is also found to be a good approximation of the reference curve.

    \begin{figure}
        \centering
        \includegraphics[width=180pt]{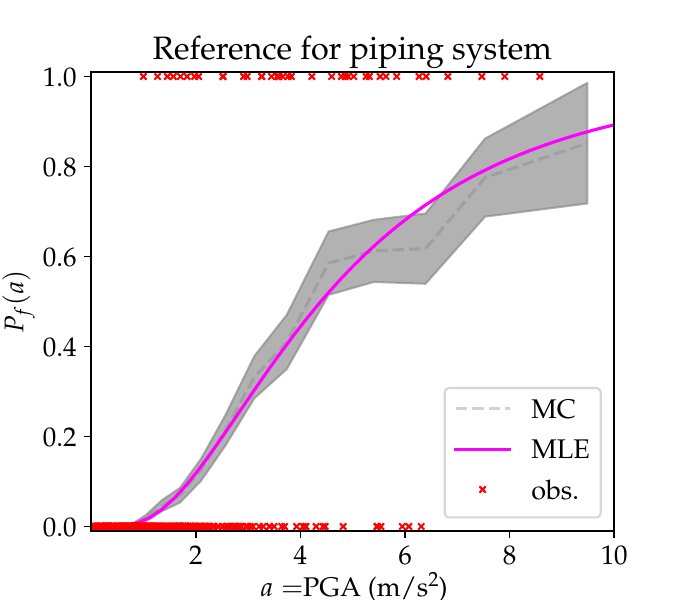}%
        \caption{Reference fragility curve $P_f^{\mathrm{MC}}$ (see Section~\ref{sec:reference}) compared with $P_f^{\mathrm{MLE}}$ (see Section~\ref{sec:metrics}) for the piping system (case study 2). Both curves were computed using the full data set generated ($10^4$ items). The red crosses represent the observations; 0 means no failure, while 1 means failure.}
        \label{fig:ref-ASG}
    \end{figure}

    Estimations similar to the first case study’s were performed here, and were found to highlight the same trends.
    As expected, for sets of $L=5000$ values of $\theta$---generated with each statistical method considered in this work---and for two sample sizes $k = 20$ and $k = 30$ of nonlinear dynamical simulations, Figure~\ref{fig:ASG-curves} shows the superiority of the Bayesian framework over the coupled MLE and bootstrap approach. Just like in the case study with the oscillator, irregularities appear with the MLE-based approach: the confidence intervals are similarly "quasi-vertical", reflecting the fact that many estimations of $\beta$ are equal to $0$. Moreover, the credibility intervals are wider with the SK prior than with the Jeffreys prior, which here too can be interpreted as an increased number of outliers of $\theta$ being generated with the SK prior. These observations are clearly supported by the results presented in Figure~\ref{fig:ASG-scatter}.
    
    For a more complete overview of their relative performances, the evaluation metrics described in Section~\ref{sec:metrics} have been computed in the same way as for the first case study: $m=200$ draws of data samples $(\mbf a^{(j)}, \mbf z^{(j)})_{j=1}^m$ have been randomly chosen to compute, for any value of $k$ ranging from $15$ to $100$, $m$ values of the metrics $\hat\cE_{L,R}^{Q|\mbf a^{(j)},\mbf z^{(j)}}$, $\hat\cS_{L,R}^{r|\mbf a^{(j)},\mbf z^{(j)}}$, $R\in\{$`MLE', `SK', `Jeffreys'$\}$, $L=5000$, $1-r=95\%$ . Their means and confidence intervals are presented in Figure~\ref{fig:ASG-errors}. These results confirm the superior performance of the Jeffreys prior compared to the other two methods.    

    \begin{figure*}
        \centering
        \includegraphics[width=130pt]{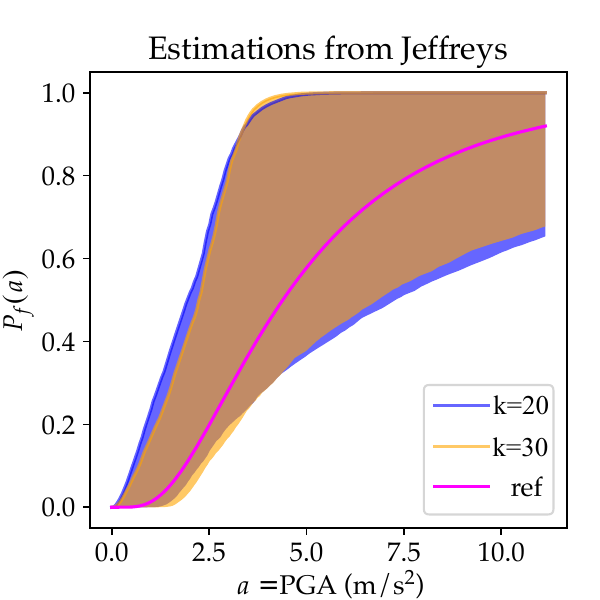}%
        \includegraphics[width=130pt]{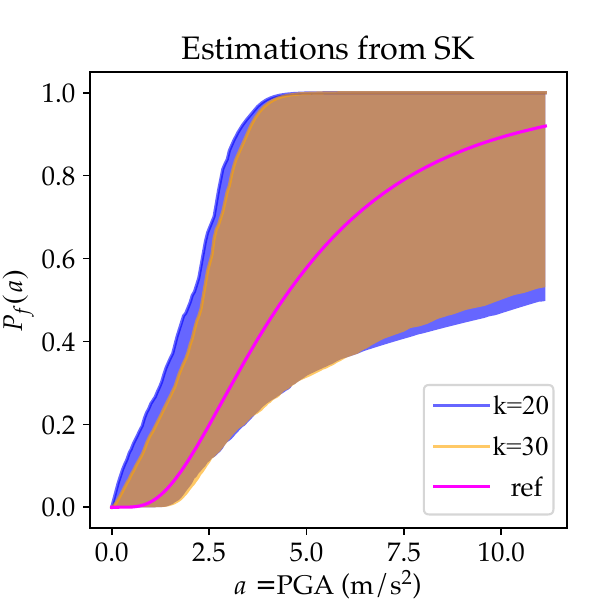}%
        \includegraphics[width=130pt]{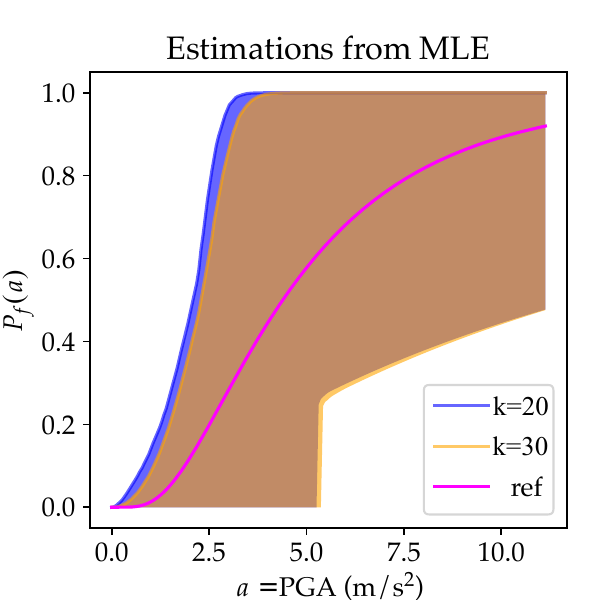}%
        \caption{$95\%$ credibility (for Bayesian estimation) or confidence (for the MLE) intervals of fragility curve estimations for the piping system, obtained from a total of $L=5000$ estimations of $\theta$ using the {statistical methods introduced in Section~\ref{sec:reference}}: (from left to right) Bayesian estimation using the Jeffreys prior, Bayesian estimation using the SK prior, and MLE with bootstrapping. For each of these, we considered two data samples {of nonlinear dynamical simulations} of two different sizes ($k=20$ in blue, $k=30$ in orange). $P_f^{\mathrm{MLE}}$ (see Section~\ref{sec:metrics}) is plotted in magenta.}
        \label{fig:ASG-curves}
    \end{figure*}

    \begin{figure*}%
        \centering%
        \includegraphics[width=130pt]{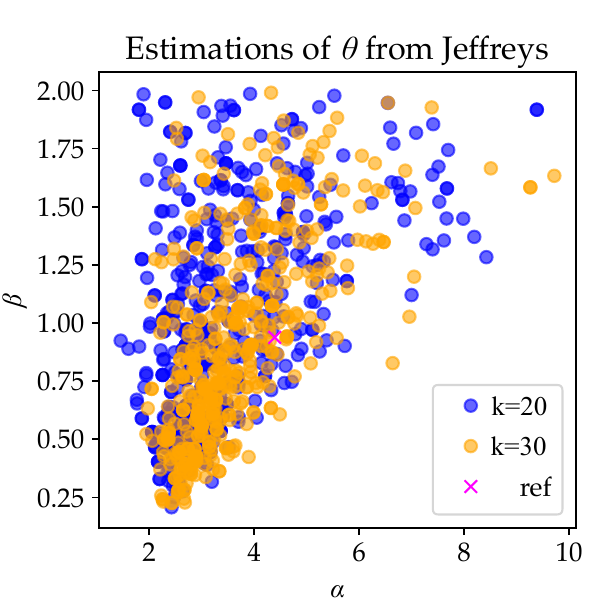}%
        \includegraphics[width=130pt]{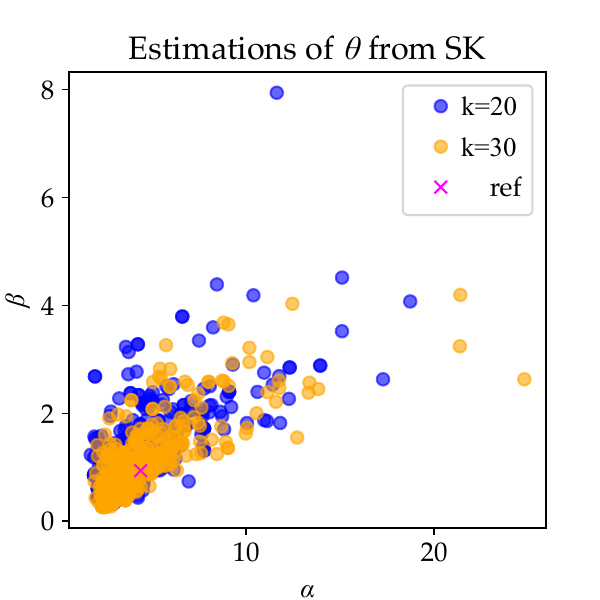}%
        \includegraphics[width=130pt]{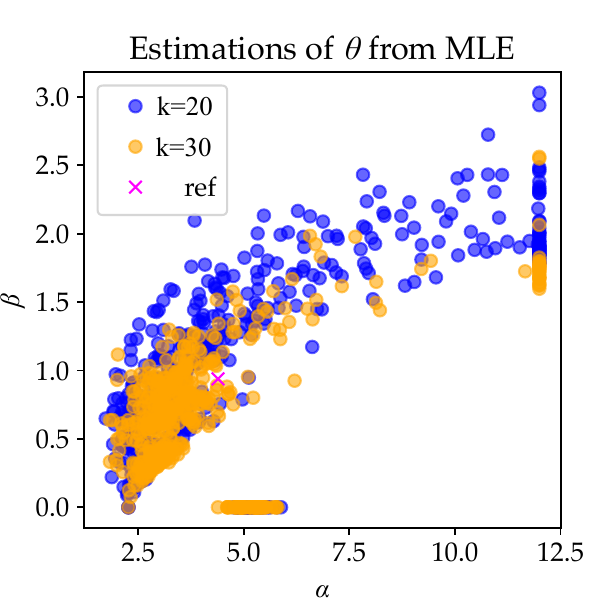}%
        \caption{Scatter plots of the generated $\theta$ for the estimation of the fragility curves presented in Figure~\ref{fig:ASG-curves} for the piping system. For all three statistical methods, we plotted $500$ points out of the $L=5000$ $\theta=(\alpha,\beta)$ estimated with two data sets of nonlinear dynamical simulations (of size $k=20$ in blue and $k=30$ in orange). The magenta crosses represent $\theta^{\mathrm{MLE}}$, used for the computation of $P_f^{\mathrm{MLE}}$ {(see Section~\ref{sec:metrics})}. This figure unveils both the outliers generated from the SK prior (center) and the irregularities characterized by null estimates of $\beta$ for the coupled MLE and bootstrap approach (right).}
        \label{fig:ASG-scatter}
    \end{figure*}

    \begin{figure*}
        \centering%
            \includegraphics[width=154pt]{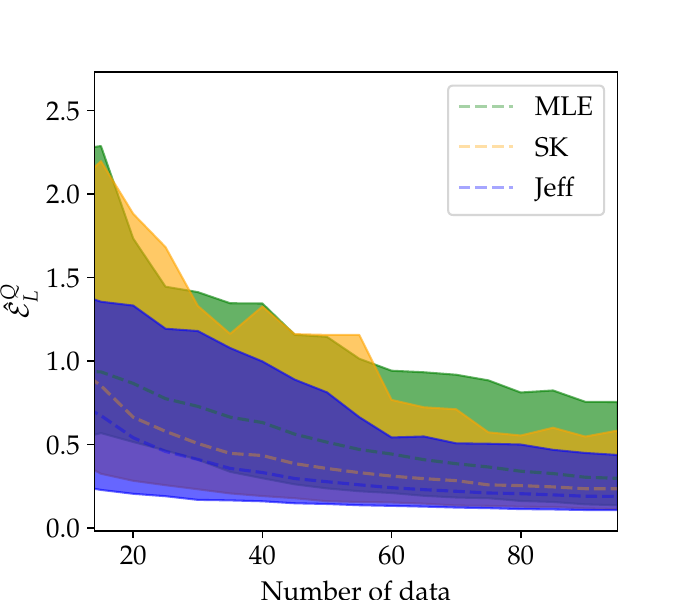} %
            \includegraphics[width=154pt]{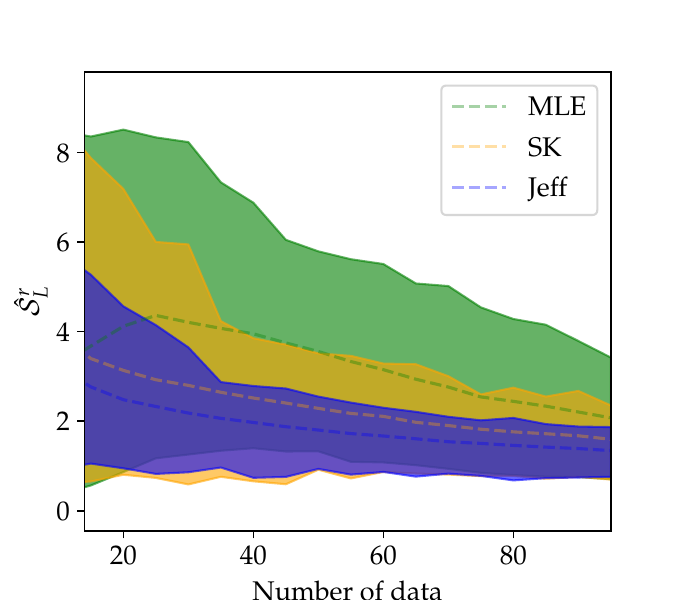}%
            \caption{Performance evaluation metrics (see Section~\ref{sec:metrics}) for the piping system computed by replications from independent draws in the full data set of nonlinear dynamical simulations and for sample sizes ranging from $k=15$ to $100$. Left: the dashed lines plot the quadratic errors as a function of the number of observations, and the shaded areas show their confidence intervals. Right: the dashed lines plot the widths of the credibility intervals (for the Bayesian estimation) or confidence intervals (for the MLE), and the shaded areas show their confidence intervals.}
        \label{fig:ASG-errors}
    \end{figure*}

\section{Conclusion \label{sec:conclusion}}

Assessing the seismic fragility of Structures and Components (SCs) is a daunting task when data is limited. The performance of the Bayesian framework in this kind of situation is well-known. Nevertheless, choosing a prior remains difficult because its impact on the \emph{a posteriori} distribution cannot be neglected, and therefore neither can its impact on the estimation of any relevant element linked to the fragility curves. 

Elaborating on the reference prior theory in order to define an objective prior, we derived, for the first time in this field of study, the Jeffreys prior for the log-normal model, with binary data that indicates the state of the structure (e.g., failure or non-failure). This prior is completely defined, and does not depend on any additional subjective choice.

This work is also an opportunity to develop a better theoretical understanding of the conditions that result in non-degenerate fragility curves (i.e., curves not taking the form of unit-step functions) in practice. This issue is quite inevitable when data is limited, since such curves are a result of the very composition of the sample. Although this issue affects every approach, we could demonstrate rigorously---i.e., both theoretically and numerically---the robustness and advantages of the proposed approach over the traditional ones found in the literature for estimating fragility curves in terms of regularization (i.e., the absence of degenerate functions when sampling fragility curves with the \emph{a posteriori} distribution) and stability (i.e., the absence of outliers when sampling the \emph{a posteriori} distribution of the parameters). {The Jeffreys prior therefore leads to more robust credibility intervals.} 

Although its numerical implementation is complex--- more so than a prior defined as the product of two classical distributions such as log-normal distributions, for instance---it is not a major issue. As a matter of fact, since it depends solely on the distribution of the IM, the ``cost'' of the initial calculation would quickly be recovered on the scale of an industrial installation containing several SCs whose fragility curves must be estimated. {For example, compared to methodologies that aim to define a prior based on mechanical calculations for a given SC, the advantage of the Jeffreys prior lies in its generic nature. The fact that it can be applied to all SCs subjected to the same seismic scenario largely compensates for the implementation of mechanical studies dedicated to each relevant SC. Additionally, the methodology can be implemented with any relevant IM without creating additional complexity.}

\section{Acknowledgments}
This research was supported by the CEA (French Alternative Energies and Atomic Energy Commission) and the SEISM Institute (www.institut-seism.fr/en/).

\appendix

\section{Prior and posterior asymptotics}\label{app:asymptotics}

This appendix is dedicated to the asymptotic study of the density functions considered in this paper. These calculations provide a proof of the proper characteristics of the posterior distributions, which are needed in order to validate the MCMC methods used for sampling. The asymptotic expansions of the Jeffreys prior can also be compared to the ones of \citeauthor{Straub2008}'s prior in order to rigorously establish their respective convergence rates. The derived asymptotic expansions of the likelihood depend on the distribution of the observed data. This gives rise to different phenomena discussed here. 
In order to carry out our analysis of the Jeffreys prior convergence rates, it is necessary to make one assumption about the IM's distribution. In this paper, we shall consider the following:
\begin{assu}\label{thm:assumptionfA}
The IM is distributed according to a log-normal distribution, i.e., there exists $\mu\in \RR$ and $\sigma\in (0,+\infty)$ such that 
    \begin{equation}
        p(a) = \frac{1}{\sqrt{2\pi\sigma^2}a}\exp\left({-\frac{(\log a-\mu)^2}{2\sigma^2}}\right).
    \end{equation}
\end{assu}
This assumption, which isn’t far from reality, feeds into the discussion conducted in Section \ref{sec:posterioriSimul}. The SK prior is therefore given by
\begin{equation}\label{eq:SKproofs}
    \pi_{SK}(\theta) = \frac{1}{\sqrt{2\pi\sigma^2}\alpha\beta} \exp\left(-\frac{(\log\alpha -\mu)^2}{2\sigma^2}\right).
\end{equation}

This appendix is organized as follows: a summary of our theoretical results about the asymptotic expansion of the likelihood is presented first in \ref{app:lieklihood-asympt}, followed by the prior density functions in \ref{app:prior-asympt}. \ref{app:posterior-discuss} contains a discussion and a comparison of the resulting posteriors’ proper characteristics. Potential scenarios resulting in degenerate fragility curves are presented in \ref{app:lik-degenarcy-discuss}. Finally, the proofs are presented in \ref{app:asympDemo}.

\subsection{Likelihood asymptotic}\label{app:lieklihood-asympt}

The following proposition gives the asymptotic behaviors of the likelihood {for different limits of $\theta = (\alpha,\beta)$}.
\begin{prop}\label{prop:likelihood}
    Let us consider $k>1$ and a data sample $(\mbf a, \mbf z)= ( (a_i)_{i=1}^k , (z_i)_{i=1}^k )$.
    Let us introduce the vectors $\mbf N=(z_i\indic_{a_i<\alpha}+(1-z_i)\indic_{a_i>\alpha})_{i=1}^k$, $\log^2\frac{\mbf a}{\alpha}=(\log^2\frac{a_i}{\alpha})_{i=1}^k$.  
    
    \begin{itemize}
        \item Fixing $\alpha>0$, then 
        \begin{equation}
        \ell_k(\mbf z|\mbf a,\theta) \conv{\beta\rightarrow\infty}2^{-k}
        \end{equation}
and
        \begin{equation}
            \ell_k(\mbf z|\mbf a,\theta) \aseq{\beta\rightarrow0}O\left(\beta^{|\mbf N|}e^{-\frac{\mbf N^T\log^2\frac{\mbf a}{\alpha}}{2\beta^2}}\right) ,
        \end{equation}
        where $|\mbf N|=\sum_{i=1}^kN_i$.
        \item Fixing $\beta>0$, then
        \begin{equation}
            \ell_k(\mbf z|\mbf a,\theta) \hspace*{-3pt} \aseq{\alpha\rightarrow0}\hspace*{-3pt} O\left(|\log\alpha|^{|\mbf z|-k} e^{-\frac{1}{2\beta^2}\sum_{i=1}^k (1-z_i)(\log a_i-\log\alpha)^2} \right)
        \end{equation}
        and 
        \begin{equation}
            \ell_k(\mbf z|\mbf a,\theta) \hspace*{-3pt}\aseq{\alpha\rightarrow\infty}\hspace*{-3pt} O\left(\log(\alpha)^{-|\mbf z|} e^{-\frac{1}{2\beta^2}\sum_{i=1}^k z_i(\log a_i-\log\alpha)^2} \right) ,
        \end{equation}
        where $|\mbf z|=\sum_{i=1}^kz_i$ is the number of failures in the observed sample.
    \end{itemize}
\end{prop}

Under general circumstances, the vector $\mbf N$ is not null, and the likelihood converges rapidly to zero when $\beta\to0$.
Under some special circumstances, however, the vector $\mbf N$ is null, and the likelihood does not converge to zero when $\beta\to0$. This happens when the failure occurrences are perfectly separated, i.e., when there exists an open interval $U$ such that $\forall\alpha\in U,$ $\forall i,$ $z_i=1\iff a_i>\alpha$. Under such circumstances, the vector $\mbf N$ is equal to ${\bf 0}$ for any $\alpha\in U$.
The likelihood converges also rapidly to zero when $\beta\to0$ when the observed sample only contains failures, resp. non-failures. Then the vector $\mbf N$ is equal to ${\bf 0}$ for $\alpha < \min(a_i)$, resp. $\alpha>\max(a_i)$.

\subsection{Prior asymptotic}\label{app:prior-asympt}

The following three propositions give the asymptotic behaviors of the Jeffreys prior $J(\theta)$ for different limits of $\theta=(\alpha,\beta)$.

\begin{prop}\label{prop:jeffb0}
    Fixing $\alpha>0$, there exists a $D'(\alpha)>0$ such that
        \begin{equation}
           J(\theta)\equi{\beta\rightarrow0} \frac{D'(\alpha)}{\beta}.
        \end{equation}
\end{prop}

\begin{prop}\label{prop:jeffbinf}
{There exists a constant $E'>0$ such that for any $\alpha>0$}
    \begin{equation}
        J(\theta)\equi{\beta\rightarrow\infty}\frac{E'}{\alpha\beta^3}.
    \end{equation}
\end{prop}

\begin{prop}\label{prop:jeffalph}
    Fixing $\beta>0$, there exists a $G''(\beta)>0$ such that
    \begin{equation}
    J(\theta) \equi{|\log\alpha|\rightarrow\infty} G''(\beta)\frac{|\log\alpha|}{\alpha}\exp\left(-\frac{(\log\alpha-\mu)^2}{2\beta^2+2\sigma^2}\right). 
\end{equation}
\end{prop}

\subsection{Discussion about posteriors}\label{app:posterior-discuss} %

\paragraph{Discussion about the proper characteristics} 
These results confirm that the Jeffreys and SK priors are not proper with respect to $\beta$. For the Jeffreys prior, the divergence and convergence rates with respect to $\beta$ only make the resulting posterior proper when the prior is coupled with the likelihood. The particular circumstances leading to a likelihood divergence when $\beta\rightarrow0$, as mentioned in \ref{app:lieklihood-asympt}, do not apply. However, one can see that this is not the case for the SK posterior, which is not integrable w.r.t. $\beta$ because of a convergence rate that is too low at $+\infty$. This prevents the validation of the MCMC estimates for this posterior, unless a truncation of the distribution is considered. This explains the generation of \emph{a posteriori} outliers using the SK prior. Note that \citeauthor{Straub2008} considered this prior in \cite{Straub2008}, but within a Bayesian framework that slightly differs from ours. In \ref{app:SKreview}, we confirmed that the posterior is not proper even when derived in the exact framework of \cite{Straub2008}.
\paragraph{Asymptotic comparison of the Jeffrey and SK priors} 
By comparing the Jeffreys and SK prior asymptotics (\ref{eq:SKproofs}), it can be observed that:
    \begin{itemize}
        \item Regarding the asymptotics w.r.t. $\beta$, while the divergence rates are the same when $\beta\to0$, the Jeffreys prior performs better when $\beta\to\infty$:
            \begin{equation*}
                J(\theta) \underset{\beta\rightarrow\infty}{\propto} \beta^{-2}\pi_{SK}(\theta).
            \end{equation*}
        Consequently, the SK posterior results in higher probabilities for high values of $\beta$ compared to the Jeffreys prior.
        \item Regarding the asymptotics w.r.t. $\alpha$, both are asymptotically close to a log-normal distribution, with a slight ``disadvantage'' for the Jeffreys prior, for which the asymptotic variance is derived by adding $\beta^2$ to the variance of the SK prior. This means that while for small values of $\beta$ (smaller than $\sigma$), both priors remain comparable w.r.t. $\alpha$, the Jeffreys prior gives higher probabilities to $\alpha$ outliers when $\beta$ also has a high value. However, as seen above, the probability for large values of $\beta$ is quite low for the Jeffreys prior compared to the SK prior. This explains why the generation of such $\alpha$ outliers has not been observed in the estimations presented in this paper.
    \end{itemize}

\subsection{On the consequences of the non-convergence towards 0 of the likelihood when $\beta$ tends towards 0 in certain circumstances}\label{app:lik-degenarcy-discuss}

Proposition \ref{prop:likelihood} formulates different convergence rates depending on how the observed data is distributed. As explained before, there are three kinds of samples which will lead to {a divergence of the likelihood} when $\beta$ tends toward $0$: (i) a sample ordered by ``no failure" and ``failure" events when classified by IM values; (ii) a sample with only ``no failure" events; and (iii) a sample with only ``failure" events. Such samples lead to unrealistic estimations of $\beta$ as $0$ within the MLE estimations, which result in unit-step fragility curves. In such circumstances, the likelihood is not fully controlled by the priors used in this work, leading to improper posteriors when the integration takes place around $\beta\to0$. This can result in degenerate fragility curves as well, yet the validity of such \emph{a posteriori} estimations would remain questionable in that case.

{
\subsection{Proofs}\label{app:asympDemo}

In this section, we will prove propositions \ref{prop:likelihood} to \ref{prop:jeffalph}.
Let us first define some upper bounds for the function $\gamma\mapsto[\Phi(\gamma)(1-\Phi(\gamma))]^{-1}$.  
\begin{lem}\label{lem:phi(1-phi)ineq1}
    For any $\gamma\in\RR$, $[\Phi(\gamma)(1-\Phi(\gamma))]^{-1}\leq4\exp\left(2\gamma^2/\pi\right)$.
\end{lem}
\begin{lem}\label{lem:phi(1-phi)ineq2} For any $\gamma\in\RR$,
    \begin{equation}
    \Phi(\gamma) (1-\Phi(\gamma)) \geq \frac{\sqrt{2/\pi}\exp(-\gamma^2/2)}{  (|\gamma|+\sqrt{\gamma^2+4}) }. 
\end{equation}
\end{lem}

\begin{proof}[Proof of lemma \ref{lem:phi(1-phi)ineq1}]
   From the following inequality about the $\erf$ function~\cite{Chu1955}:
\begin{equation*}
        \forall \gamma>0,\, \sqrt{1-e^{-\frac{\gamma^2}{2}}}\leq\erf(\gamma/\sqrt{2})\leq\sqrt{1-e^{-2\frac{\gamma^2}{\pi}}},
    \end{equation*}
we can deduce that, for any $\gamma>0$,
\begin{align*}
    e^{-\frac{2\gamma^2}{\pi}}\hspace*{-3pt}&\leq 1-\erf(\gamma/\sqrt{2})^2\leq e^{-\frac{\gamma^2}{2}} ,  \\
    \frac{1}{4}e^{-\frac{2\gamma^2}{\pi}}\hspace*{-3pt}&\leq\frac{1}{4}\left(1-\erf\left({\gamma}/{\sqrt{2}}\right)\right)\left(1+\erf\left({\gamma}/{\sqrt{2}}\right)\right)\leq\frac{1}{4} e^{-\frac{\gamma^2}{2}}  , 
\end{align*}
the middle term being equal to $\Phi(\gamma)(1-\Phi(\gamma))$. This implies that:
\begin{equation*}
    [\Phi(\gamma)(1-\Phi(\gamma))]^{-1}\leq 4 e^{\frac{2\gamma^2}{\pi}} ,
\end{equation*}
hence the result for $\gamma>0$.

While it is clear that the inequality still stands for $\gamma=0$, notice that from $\Phi(-\gamma)=1-\Phi(\gamma)\,\forall\gamma\in\RR$ it follows that $\gamma\mapsto\Phi(\gamma)(1-\Phi(\gamma))$ is an even function. Thus, the inequality still stands for any $\gamma<0$; this concludes the proof of the lemma.
\end{proof}

\begin{proof}[Proof of lemma \ref{lem:phi(1-phi)ineq2}]
    Komatsu's inequality \cite[p.~17]{itomckean}:
    \begin{equation*}
        \forall \gamma>0,\, \frac{2}{\sqrt{\gamma^2+4}+\gamma}\leq e^{\frac{\gamma^2}{2}}\hspace*{-4pt}\int_\gamma^\infty\hspace*{-4pt} e^{-\frac{t^2}{2}}dt\leq \frac{2}{\sqrt{\gamma^2+2}+\gamma}
    \end{equation*}
    implies
    \begin{equation*}
        \forall \gamma>0,\, \frac{2e^{-\frac{\gamma^2}{2}}}{\sqrt{\gamma^2+4}+\gamma}\leq\hspace*{-3.5pt} \sqrt{2\pi}(1-\Phi(\gamma)) \leq\hspace*{-2pt} \frac{2e^{-\frac{\gamma^2}{2}}}{\sqrt{\gamma^2+2}+\gamma} .
    \end{equation*}
    Since $0<\Phi<1$ it follows for $\gamma>0$ that:
    \begin{align*}
        \Phi(\gamma)(1-\Phi(\gamma)) &\geq \frac{2e^{-\frac{\gamma^2}{2}}}{\sqrt{\gamma^2+4}+\gamma}\left(1-\frac{2e^{-\frac{\gamma^2}{2}}}{\sqrt{\gamma^2+2}+\gamma}\right)
        \\   &
            \geq \frac{\sqrt{2/\pi}e^{-\frac{\gamma^2}{2}}}{\sqrt{\gamma^2+4}+\gamma}.
    \end{align*}
    Finally, since $\Phi(-\gamma)=1-\Phi(\gamma)$, $\gamma\mapsto\Phi(\gamma)(1-\Phi(\gamma))$ is an even function and we obtain for any $\gamma\in\RR$
    \begin{equation*}
        \Phi(\gamma)(1-\Phi(\gamma))\geq \frac{\sqrt{2/\pi}e^{-\frac{\gamma^2}{2}}}{\sqrt{\gamma^2+4}+|\gamma|}.
    \end{equation*}
\end{proof}

}

    \subsubsection{Proof of Proposition \ref{prop:likelihood}}

As a reminder, the likelihood is expressed as:
    \begin{align*}%
        &\ell_k(\mbf z|\mbf a,\theta)\\  &= \prod_{i=1}^k\Phi\left(\frac{\log a_i-\log\alpha}{\beta}\right)^{z_i}\left(1-\Phi\left(\frac{\log a_i-\log\alpha}{\beta}\right)\right)^{1-z_i} \\
            & = \exp\left[\sum_{i=1}^k\left(z_i\log\Phi(\gamma_i) + (1-z_i)\log(1-\Phi(\gamma_i))\right)\right] ,
    \end{align*}
denoting $\gamma_i=\beta^{-1}\log\frac{a_i}{\alpha}$. %

To treat the case where $\beta\to\infty$ we can observe that while $\alpha$ is fixed, the quantities $\Phi(\gamma_i)$ all converge to ${1}/{2}$. The product of those limits gives the limit $\ell_k(\mbf z|\mbf a,\theta)\conv{\beta\rightarrow\infty}2^{-k}$.

For the other cases, it should be reminded that $\Phi(x)=\frac{1}{2}(1+\erf(x/\sqrt{2}))$ and $\erf(x)\aseq{x\rightarrow\infty}1-\frac{e^{-x^2}}{x\sqrt{\pi}} + o\left(\frac{e^{-x^2}}{x}\right)$, which leads to 
    \begin{equation}\label{eq:Jasymp:phiasymp}
        \Phi(x)\aseq{x\rightarrow\infty}1 - \frac{e^{-\frac{x^2}{2}}}{x\sqrt{2\pi}} + o\left(\frac{e^{-\frac{x^2}{2}}}{x}\right).
    \end{equation}
Let us consider an $i\in\{1,\dots,k\}$ and compute
    \begin{align*}
       & z_i\log\Phi(\gamma_i) + (1-z_i)\log(1-\Phi(\gamma_i)) \\
          &  \aseq{\gamma_i\rightarrow\infty} z_i\log\left(1-\frac{e^{-\frac{\gamma_i^2}{2}}}{\gamma_i\sqrt{2\pi}} + o\left(\frac{e^{-\frac{\gamma_i^2}{2}}}{\gamma_i}\right)\right)\\
          &\hspace*{2em}+(1-z_i)\log\left(\frac{e^{-\frac{\gamma_i^2}{2}}}{\gamma_i\sqrt{2\pi}} + o\left(\frac{e^{-\frac{\gamma_i^2}{2}}}{\gamma_i}\right)\right)\\
           & \aseq{\gamma_i\rightarrow\infty} -z_i\frac{e^{-\frac{\gamma_i^2}{2}}}{\gamma_i\sqrt{2\pi}}  %
                +(1-z_i)\log\left(\frac{e^{-\frac{\gamma_i^2}{2}}}{\gamma_i\sqrt{2\pi}}\right) + o(1)\\
           & \aseq{\gamma_i\rightarrow\infty} - (1-z_i)\frac{\gamma_i^2}{2} - (1-z_i)\log(\gamma_i\sqrt{2\pi}) + o(1).
    \end{align*}
Using the relation $\Phi(-x)=1-\Phi(x)$, it follows that calculation to the following, in the case $a_i<\alpha$:
    \begin{align*}
       & z_i\log\Phi(\gamma_i) + (1-z_i)\log(1-\Phi(\gamma_i))\\
        &    \aseq{\gamma_i\rightarrow-\infty} - z_i\frac{\gamma_i^2}{2} - z_i\log(-\gamma_i\sqrt{2\pi}) + o(1).
    \end{align*}

{
Going back to the likelihood asymptotics, let us first $\alpha>0$ and suppose $\beta\to0$. Thus, denoting the vectors $\mbf N=(z_i\indic_{a_i<\alpha}+(1-z_i)\indic_{a_i>\alpha})_{i=1}^k$ and 
$\log^2\frac{\mbf a}{\alpha}=(\log^2\frac{a_i}{\alpha})_{i=1}^k$, we obtain
    \begin{align*}
        \ell_k(\mbf z|\mbf a,\theta) &\aseq{\beta\rightarrow0} \frac{C(\alpha)}{\sqrt{2\pi}^{|\mbf N|}}\beta^{|\mbf N|} e^{-\frac{\mbf N^T\log^2\frac{\mbf a}{\alpha}}{2\beta^2}+o(1)}\\ %
            &\aseq{\beta\rightarrow0}O\left(\beta^{|\mbf N|}e^{-\frac{\mbf N^T\log^2\frac{\mbf a}{\alpha}}{2\beta^2}}\right) ,
    \end{align*}
where $|\mbf N|=\sum_{i=1}^kN_i$ and $C(\alpha) =\prod_{i=1}^k\left|\log\frac{a_i}{\alpha}\right|^{N_i}$. %

Let us then fix $\beta>0$ to get
    \begin{align*}
        &\ell_k(\mbf z|\mbf a,\theta)\\ &\aseq{\alpha\rightarrow\infty} \frac{\beta^{|\mbf N|}}{\sqrt{2\pi}^{|\mbf N|}} \left(\prod_{i=1}^k\log\left(\frac\alpha{a_i}\right)^{-z_i}\right) \\
        &\hspace*{1.5em}\cdot \exp\left({-\frac{1}{2\beta^2}\sum_{i=1}^k z_i(\log a_i-\log\alpha)^2+o(1)}\right)
        \\
            &\aseq{\alpha\rightarrow\infty}O\left(\log(\alpha)^{-|\mbf z|} e^{-\frac{1}{2\beta^2}\sum_{i=1}^k z_i(\log a_i-\log\alpha)^2} \right),
    \end{align*}
where $|\mbf z|=\sum_{i=1}^kz_i$ is the number of failures in the observed sample.
Similarly, we obtain
    \begin{align*}
        &\ell_k(\mbf z|\mbf a,\theta) \\
        &\aseq{\alpha\rightarrow0}
            \frac{\beta^{|\mbf N|}}{\sqrt{2\pi}^{|\mbf N|}} \left(\prod_{i=1}^k\log\left(\frac{a_i}\alpha\right)^{-(1-z_i)}\right) \\
            &\hspace*{1.5em}\cdot \exp\left({-\frac{1}{2\beta^2}\sum_{i=1}^k (1-z_i)(\log a_i-\log\alpha)^2+o(1)}\right)
        \\
            &\aseq{\alpha\rightarrow0}O\left(|\log\alpha|^{|\mbf z|-k} e^{-\frac{1}{2\beta^2}\sum_{i=1}^k (1-z_i)(\log a_i-\log\alpha)^2} \right).
    \end{align*}
}

    \subsubsection{Proof of proposition \ref{prop:jeffb0}}

Let $\alpha>0$. For $0\leq k\leq2$, let us consider $A_{k1,k2}=A_{k1}+A_{k2}$ with $A_{kj}$ defined in (\ref{eq:Aij}):
$$
A_{k1,k2} =  \int_{0}^{\infty}\log^k\frac{a}{\alpha}\frac{\Phi'(\gamma(a))^2}{\Phi(\gamma(a))(1-\Phi(\gamma(a)))}p(a)da .
$$
We obtain
    \begin{align}
        \label{eq:AkAk}
        A_{k1,k2} &=   \beta^{k+1}\int_{-\infty}^\infty  F_{A_{k1,k2}}(\gamma)d\gamma ,\\
        F_{A_{k1,k2}}(\gamma) &=  \frac{\gamma^k}{2\sqrt{\pi^3\sigma^2}} \frac{e^{-\gamma^2}e^{-\frac{(\beta\gamma-\mu+\log\alpha)^2}{2\sigma^2}}}{\Phi(\gamma)(1-\Phi(\gamma))} .
    \end{align}
Using lemma \ref{lem:phi(1-phi)ineq1} an upper bound can be derived for $F_{A_{k1k2}}$: for any $\gamma\in\RR,\,\beta>0$,
\begin{equation}\label{eq:Jasymp:majFkk}
    |F_{A_{k1,k2}}(\gamma)|\leq \tilde C|\gamma|^ke^{-\frac{1}{3}\gamma^2}
\end{equation}
which defines an integrable function on $\RR$, $\tilde C$ being a constant independent of $\gamma$ and $\beta$. Hence the limit
\begin{equation*}
    \lim_{\beta\rightarrow0}\int_{-\infty}^\infty F_{A_{k1,k2}}(\gamma)d\gamma =\hspace*{-3pt} \int_{-\infty}^\infty \frac{\gamma^k}{2\sqrt{\pi^3\sigma^2}}\frac{e^{-\gamma^2}e^{-\frac{(\mu-\log\alpha)^2}{2\sigma^2}}}{\Phi(\gamma)(1-\Phi(\gamma))}d\gamma.
\end{equation*}
{%
The last integral is null when $k=1$ as the integrand is odd in this case. 
When $k$ is even, the integrand is positive valued almost everywhere, which implies that the integral is positive. 
From this, we can establish that  $A_{k1,k2}\equi{\beta\rightarrow0}D_k(\alpha)\beta^{k+1}$ for some $D_k(\alpha)>0$ if $k=0,2$, and $A_{11,12}\aseq{\beta\rightarrow0}o(\beta^2)$.}

Looking back at the Fisher information matrix, we can state that
    \begin{equation*}
        \det\cI(\theta)\aseq{\beta\rightarrow0}\frac{D_0(\alpha)D_2(\alpha)}{\alpha^4\beta^2} + o\left(\frac{1}{\beta^2}\right).
    \end{equation*}
Finally, we obtain:
\begin{equation}\label{eq:Jrate}
    J(\theta)\equi{\beta\rightarrow0} \frac{D'(\alpha)}{\beta},
\end{equation}
where $D'(\alpha)>0$ is a constant independent of $\beta$.

\subsubsection{Proof of proposition \ref{prop:jeffbinf}}

As a reminder, the asymptotic expansion of the $\erf$ function in $0$ is:
    \begin{equation*}
        \erf(\gamma)\aseq{\gamma\rightarrow0}\frac{2}{\sqrt{\pi}} \gamma+O(\gamma^2) ,
    \end{equation*}
which allows us to state the behavior of $\Phi(\gamma)$ when $\gamma\to0$: 
\begin{equation*}
    \Phi(\gamma)\aseq{\gamma\rightarrow0} \frac{1}{2} + \frac{1}{\sqrt{2\pi}} \gamma+O(\gamma^2).    
\end{equation*}
 
Let us now fix $\alpha>0$ and consider $A_{k1,k2}= A_{k1}+A_{k2}$:
\begin{align*}
   A_{k1,k2} &= \int_{-\infty}^\infty \tilde F_{A_{k1,k2}}(x)dx, \\
    \tilde F_{A_{k1,k2}}(x) &= \frac{x^k}{2\sqrt{\pi^3\sigma^2}}\frac{e^{-\frac{x^2}{\beta^2}}e^{-\frac{(x-\mu+\log\alpha)^2}{2\sigma^2}}}{\Phi(\beta^{-1}x)(1-\Phi(\beta^{-1}x))}.
\end{align*}
Let us note the convergence of $\tilde F_{A_{k1,k2}}(x)$ towards %
an integrable function when $\beta\to\infty$.
Moreover, lemma \ref{lem:phi(1-phi)ineq1} allows us to define bounds for $\tilde F_{A_{k1,k2}}$:
\begin{align*}
    |\tilde F_{A_{k1,k2}}(x)|&\leq \frac{2|x|^k}{\sqrt{\pi^3\sigma^2}}e^{-\frac{(x-\mu+\log\alpha)^2}{2\sigma^2}} e^{2\frac{x^2}{\pi\beta^2}} \\
        &\leq \frac{2|x|^k}{\sqrt{\pi^3\sigma^2}} e^{-\frac{(x^2-2(\mu-\log\alpha))^2}{4\sigma^2}} e^{\frac{(\mu+\log\alpha)^2}{2\sigma^2}} ,
\end{align*}
for any $x\in\RR$ and $\beta>2\sigma/\sqrt{\pi}$. This dominating function is integrable on $\RR$. Thus, when $\beta\to\infty$, $A_{k1,k2}$ admits a limit expressed by:
\begin{align*}
    \lim_{\beta\rightarrow\infty}A_{k1,k2} &= E_k(\alpha) \\ &= \int_{-\infty}^\infty\frac{2x^k}{\sqrt{\pi^3\sigma^2}}e^{-\frac{(x-\mu+\log\alpha)^2}{2\sigma^2}} dx= \frac{2\sqrt{2}}{\pi}\EE[X^k] ,
\end{align*}
with $X\sim\cN(\mu-\log\alpha,\sigma^2)$.
Recalling the expression of the Jeffreys prior:
    \begin{equation*}
        J(\theta)^2 = \left|\frac{1}{\alpha^2\beta^6}A_{01,02}A_{21,22} - \frac{1}{\alpha^2\beta^6}A_{11,12}^2\right| ,
    \end{equation*}
    we can deduce that it is equivalent to $(E_0(\alpha)E_2(\alpha)-E_1^2(\alpha))/\alpha^2\beta^6$ when $\beta\to\infty$.
    Finally,
    \begin{equation*}
        J(\theta)\equi{\beta\rightarrow\infty}\frac{E'}{\alpha\beta^3} ,
    \end{equation*}
with $E'=\sqrt{E_0(\alpha)E_2(\alpha)-E_1^2(\alpha)} = 2 \sigma /\pi$.

    \subsubsection{Proof of proposition \ref{prop:jeffalph}}
    
As a preliminary result, let us use Eq.~(\ref{eq:Jasymp:phiasymp}) to obtain
\begin{equation}\label{eq:Jasymp:Phi(1-Phi)equi}
    \Phi(\gamma)(1-\Phi(\gamma))\equi{|\gamma|\to\infty} \frac{e^{\frac{-\gamma^2}{2}}}{|\gamma|\sqrt{2\pi}}.
\end{equation}

Let us consider $A_{k1,k2}=A_{k1}+A_{k2}$:
\begin{align*}
    & A_{k1,k2} \\ &= C'\int_0^{\infty}\left(\log\frac{a}{\alpha}\right)^k\frac{e^{-\frac{1}{\beta^{2}}\log^2\frac{a}{\alpha}}e^{-\frac{(\log a-\mu)^2}{2\sigma^2}}}{\Phi(\beta^{-1}\log\frac{a}{\alpha})\left(1-\Phi(\beta^{-1}\log\frac{a}{\alpha})\right)}\frac{da}{a},
\end{align*}
denoting $C'=\sqrt{4\pi^3\sigma^2}^{-1}$. By substituting
\begin{align*}
    \nu &= \log a - \frac{\sigma^2}{\sigma^2+ \beta^2}\log\alpha - \frac{\beta^2}{\sigma^2+\beta^2}\mu \\  & = \log a -r\log\alpha - s\mu,
\end{align*}
we obtain
\begin{align*}%
    & A_{k1,k2} = C'\int_{-\infty}^\infty \hat F_{A_{k1,k1}}(\nu)d\nu , 
\end{align*}
and
\begin{align*}
    &\hat F_{A_{k1,k1}}(\nu)\\ &=   (\nu+(r-1)\log\alpha +s\mu)^k\frac{e^{-\frac{(\nu + (r+1)\log\alpha +s\mu)^2 }{\beta^2} } e^{-\frac{(\nu +r\log\alpha+ (s-1)\mu)^2}{2\sigma^2}}}{\left[\Phi(1-\Phi)\right](h_\beta(\nu))} ,
\end{align*}
where $h_\beta(\nu)=\beta^{-1}(\nu+(r-1)\log\alpha+s\mu)$. Using Eq.~(\ref{eq:Jasymp:Phi(1-Phi)equi}), we obtain
\begin{equation}\label{eq:Jasymp:phi1phinurs}
    \left[\Phi(1-\Phi)\right](h_\beta(\nu))\hspace*{-2pt} \equi{|\log\alpha|\rightarrow\infty} \hspace*{-1pt}\frac{\beta e^{-\frac{(\nu+(r-1)\log\alpha+s\mu)^2}{2\beta^2}}}{|\nu + (r-1)\log\alpha + s\mu|\sqrt{2\pi}}.
\end{equation}
Then for a clear understanding of the asymptotic behavior of $\hat F_{A_{k1,k2}}$, let us compute 
\begin{align}
\nonumber
     -&\frac{(\nu + (r-1)\log\alpha +s\mu)^2 }{2\beta^2} - \frac{(\nu +r\log\alpha+ (s-1)\mu)^2}{2\sigma^2} \\
\nonumber
=& -\left(\frac{1}{2\beta^2}+\frac{1}{2\sigma^2}\right)\nu^2 \\
&-\left(\frac{((r-1)\log\alpha+s\mu)^2}{2\beta^2} +\frac{(r\log\alpha+(s-1)\mu)^2}{2\sigma^2}\right)\nonumber\\
         =& -\left(\frac{1}{2\beta^2}+\frac{1}{2\sigma^2}\right)\nu^2 - \frac{(\log\alpha - \mu )^2}{2(\beta^2+\sigma^2)}.
\label{eq:Jasymp:nurs}
\end{align}
Expanding $\hat F_{k1,k2}(\nu) = \sum_{j=1}^kC_k^j(r-1)^j\log^j\alpha(\nu+s\mu)^{k-j}g(\nu)$ $ = \sum_{j=0}^k\hat F_{k1,k2}^{(j)}(\nu)$, with $g(\nu)$ defined as    \begin{equation*}
        g(\nu) = \frac{e^{-\frac{(\nu + (r+1)\log\alpha +s\mu)^2 }{\beta^2} } e^{-\frac{(\nu +r\log\alpha+ (s-1)\mu)^2}{2\sigma^2}}}{\left[\Phi(1-\Phi)\right](\beta^{-1}(\nu + (r-1)\log\alpha + s\mu))}.
    \end{equation*}
By combining Eqs.~(\ref{eq:Jasymp:phi1phinurs}) and (\ref{eq:Jasymp:nurs}), we obtain that $\hat F_{k1,k2}^{(j)}$ satisfies
    \begin{align*}
         &\hat F_{k1,k2}^{(j)}(\nu)e^{\frac{(\log\alpha-\mu)^2}{2\beta^2+2\sigma^2}}(\log\alpha)^{-j}|\log\alpha|^{-1} \\
&            \conv{|\log\alpha|\rightarrow\infty} (\sqrt{2\pi}\beta)^{-1}C_k^j(r-1)^{j+1}(\nu+s\mu)^{k-j} e^{-\left(\frac{1}{2\beta^2}+\frac{1}{2\sigma^2}\right)\nu^2}.
    \end{align*}
Using lemma \ref{lem:phi(1-phi)ineq2}, we can also define an upper bound for the above function in the form of an integrable function expressed as
\begin{align*}
    |\hat F_{k1,k2}^{(j)}(\nu)|e^{\frac{(\log\alpha-\mu)^2}{2\beta^2+2\sigma^2}}|\log\alpha|^{-j-1} &\leq \frac{\sqrt{2/\pi}e^{-\left(\frac{1}{2\beta^2}+\frac{1}{2\sigma^2}\right)\nu^2}}{ |h_\beta(\nu)|+\sqrt{h_\beta(\nu)^2+4}}\\
       & \leq\frac{1}{\sqrt{2\pi}}e^{-\left(\frac{1}{2\beta^2}+\frac{1}{2\sigma^2}\right)\nu^2}.
\end{align*}
Therefore, we can switch the limits and the integration, and the following results are obtained:
\begin{align*}
    C''\beta A_{01,02}e^{\frac{(\log\alpha-\mu)^2}{2\beta^2+2\sigma^2}} \aseq{\log\alpha\rightarrow\infty}& (1-r)G\log\alpha - s\mu G  +o(1),
\end{align*}
\begin{align*}
    C''\beta A_{11,12}e^{\frac{(\log\alpha-\mu)^2}{2\beta^2+2\sigma^2}}
        &\aseq{\log\alpha\rightarrow\infty} -(1-r)^2G\log^2\alpha - s^2\mu^2 G \\
        &+ 2(1-r)s\mu G\log\alpha - G' + o(1),
\end{align*}
\begin{align*}
&C''\beta A_{22,22}e^{\frac{(\log\alpha-\mu)^2}{2\beta^2+2\sigma^2}}\\
&\aseq{\log\alpha\rightarrow\infty}\hspace*{-4pt}
(1-r)^3G\log^3\alpha - s^3\mu^3 G 
-3(1-r)^2s\mu G\log^2\alpha \\
&\hspace*{2em}-3(1-r)s^2\mu^2 G\log\alpha+3(1-r)G'\log\alpha  + o(1),
\end{align*}
with $C''=(C'\sqrt{2\pi})^{-1}$, $G=\sigma\beta\sqrt{2\pi(\beta^2+\sigma^2)^{-1}}$ and $G'=G^2/2\pi$. 
This way
\begin{align*}
    &C''\alpha\beta^{8} J(\theta)^2e^{\frac{(\log\alpha-\mu)^2}{2\beta^2+2\sigma^2}}\\
      &  \aseq{\log\alpha\rightarrow\infty}\hspace*{-4pt}  
        3(r-1)^2s^2\mu^2G^2\log^2\alpha 
        + 3(r-1)^2s^2\mu^2G^2\log^2\alpha \\
        &\hspace*{2em} +3 (r-1)^2GG'\log^2\alpha 
        - 4(r-1)^2s^2\mu^2G^2\log^2\alpha \\
        &\hspace*{2em} - 2(r-1)^2s^2\mu^2G^2\log^2\alpha - 2(r-1)^2GG'\log^2\alpha \\
        &\hspace*{2em}+ o(\log^2\alpha).
\end{align*}
Note that the above equality is still valid when $\log\alpha\to-\infty$. Finally
\begin{equation*}
    J(\theta) \equi{|\log\alpha|\rightarrow\infty} G''(\beta)\frac{|\log\alpha|}{\alpha}\exp\left(-\frac{(\log\alpha-\mu)^2}{2\beta^2+2\sigma^2}\right)  ,
\end{equation*}
with 
\begin{align*}
    G''(\beta) &=C''^{-1}(r-1)^2GG'\beta^{-4} 
        = \frac{2\sigma^3\beta^3 }{\sqrt{\pi}(\sigma^2+\beta^2)^{7/2}}. %
\end{align*}

\section{{A review of the properties of the SK posterior}}\label{app:SKreview}

In this paper, we have compared our approach with the one that results from an adaptation of the prior suggested by \citet{Straub2008}. We proved in \ref{app:asymptotics} that within our framework, this prior results in an improper posterior. This puts the validity of the MCMC estimations into question, and could explain the lower performance of the SK prior compared to the Jeffreys prior. 
In \cite{Straub2008}, the authors use the Bayesian methodology the same way we do, yet the consideration of uncertainties over the observed earthquake intensity measures and the equipment capacities leads to a slightly different likelihood.
In order to verify that the drawbacks of their prior highlighted in this paper are not due to our statistical choices, we dedicated this appendix to the study of the asymptotic expansions of the posterior in the exact framework presented in~\cite{Straub2008}.
We shall first introduce the exact model of \citeauthor{Straub2008} for the estimation of seismic fragility curves in \ref{subapp:C1}, using notations consistent with our study. We will then derive the likelihood and its asymptotics in \ref{subapp:C2}. Finally, we will express the convergence rates of the posterior in \ref{subapp:C3}, which will allow us to conclude that the SK posterior is indeed improper.

\subsection{Statistical model and likelihood}
\label{subapp:C1}
Let us consider the observations of earthquakes labeled $l=1,\dots,L$ at equipment labeled $i=1,\dots,I_j$ located in substations labeled $j=1,\dots,J$.
The observed items are $(\mbf z_{jl},\hat a_{jl})_{j,l}$, $\mbf z_{jl}=(z_{ijl})_i$ being the failure occurrences of the $I_j$ pieces of equipment at substation $j$ during earthquake $l$ ($z_{ijl}\in\{0,1\}$), and $\hat a_{jl}$ being the observed IM at substation $j$ during earthquake $l$ ($\hat a_{jl}\in(0,+\infty)$).
They are assumed to follow the latent model presented below.

At substation $j$ the $l$-th earthquake results in an IM value $a_{jl}$ that is observed with an uncertainty multiplicative noise: $\log\hat a_{jl}=\log a_{jl}+\eps_{jl}$ where $\eps_{jl}\sim\cN(0,\sigma_\eps^2)$. The noise variance $\sigma_\eps^2$ is supposed to be known. %
The uncertain intrinsic capacity of equipment $i$ at substation $j$ is $r_{ij}\sim\cN(\mu_r,\sigma^2_r)$ and $y_{jl}\sim\cN(0,\sigma^2_y)$ is the uncertain factor common to all equipment capacities at substation $j$ during earthquake $l$.
The random variables $r_{ij}$, $y_{jl}$ and $\eps_{jl}$ are supposed to be independent.

A failure of equipment $i$ at substation $j$ during earthquake $l$ is considered when the performance of the structural component $g_{ijl}$ satisfies $g_{ijl}>0$.%
This performance can be expressed as 
    \begin{equation*}
      g_{ijl} = \log \hat a_{jl}+\eps_{jl}-y_{jl}-r_{ij}=x_{jl}-r_{ij}  
    \end{equation*}
with $x_{jl} = \log \hat a_{jl}+\eps_{jl}-y_{jl}$.

This establishes the following conditional relation between the observed data:
    \begin{equation}
        p(z_{ijl}|\hat a_{jl},\Sigma) =\hspace*{-3pt} \int_\RR p(z_{ijl}|x_{jl},\hat a_{jl},\Sigma)\frac{\exp\left({-\frac{(x_{jl}-\log\hat{a}_{jl})^2}{2(\sigma^2_\eps+\sigma^2_y)}}\right)}{\sqrt{2\pi(\sigma_\eps^2+\sigma_y^2)}}dx_{jl} ,
    \end{equation}
denoting $\Sigma=(\sigma_r,\,\sigma_y,\,\mu_r)$, and with 
    \begin{equation}\label{eq:SKr:condzx}
        p(z_{ijl}|x_{jl},\hat a_{jl},\Sigma) = \Phi\left(\frac{x_{jl}-\mu_r}{\sigma_r}\right)^{z_{ijl}}\left(1-\Phi\left(\frac{x_{jl}-\mu_r}{\sigma_r}\right)\right)^{1-z_{ijl}} \hspace*{-4pt} ,
    \end{equation}
when substation $j$ is only affected by one earthquake. The method proposed in \cite{Straub2008} actually considers cases in which a substation may be impacted by two successive earthquakes and takes into account the fact that its response to the second would be correlated to its response to the first one. This would lead to a different likelihood. However, it is mentioned that this possibility only concerns a small number of data points. We can therefore limit our calculations to the simplest case and assume $l=L=1$. The subscript $l$ will therefore be dropped in what follows.

Finally, the likelihood for this model can be expressed as:
    \begin{align}
    \label{eq:SKr:likelihood}
        &\ell_J(\mbf z | \mbf{\hat a}, \Sigma)\\
            &= \prod_{j=1}^J \int_\RR\prod_{i=1}^{I_j} p(z_{ij}|x_{j},\log\hat a_{j},\Sigma) \frac{\exp\left({-\frac{(x_{j}-\log\hat{a}_{j})^2}{2(\sigma^2_\eps+\sigma^2_y)}}\right)}{\sqrt{2\pi(\sigma_\eps^2+\sigma_y^2)}}dx_{j}  ,\nonumber
    \end{align}
denoting $\mbf z=(\mbf z_j)_{j=1}^J$, $\mbf{\hat a}=(\hat a_j)_{j=1}^J$, and with the integrated conditional distribution defined in Eq.~(\ref{eq:SKr:condzx}).

    In the Bayesian framework introduced in \cite{Straub2008}, the model parameter is $\Sigma$. Let us denote 
    $\alpha=\exp\mu_r$, $\beta=\sqrt{\sigma_r^2+\sigma^2_y}$ and $\rho=\sigma^2_y/\beta^2$. Denoting $\theta=(\alpha,\beta,\rho)$, the knowledge of $\Sigma$ then becomes equivalent to the one of $\theta$ and the likelihood of Eq.~(\ref{eq:SKr:likelihood}) can be expressed conditionally to $\theta$ instead of $\Sigma$:

    \begin{align}
    \label{eq:SKr:likelihood-theta}
        &\ell_J(\mbf z | \mbf{\hat a}, \theta)\\
            &= \prod_{j=1}^J  \int_\RR\prod_{i=1}^{I_j} %
            \Psi^{z_{ij}}\left(\frac{x-\log\alpha}{\beta\sqrt{1-\rho}}\right)
                \frac{\exp\left({-\frac{(x-\log\hat{a}_{j})^2}{2(\sigma^2_\eps+\rho\beta^2)}}\right)}{\sqrt{2\pi(\sigma_\eps^2+\rho\beta^2)}}dx ,\nonumber
    \end{align}
    where the notation $\Psi^{z_{ij}}(\gamma)$ is used to denote $\Phi(\gamma)^{z_{ij}}(1-\Phi(\gamma))^{1-z_{ij}}$.

    \citeauthor{Straub2008} propose the following improper prior distribution for the parameter $\theta$:
    \begin{equation}\label{eq:SKr:SKprior}
        \pi_{SK}(\theta) \propto \frac{1}{\beta\alpha}\exp\left(-\frac{(\log\alpha-\mu)^2}{2\sigma^2}\right)\indic_{0\leq\rho\leq1}.
    \end{equation}
    \emph{A posteriori} estimations of $\theta$ are consequently generated from MCMC methods %
    \begin{equation}\label{eq:SKr:post}
        p(\theta|\mbf z,\mbf{\hat a})\propto \ell_J(\mbf z | \mbf{\hat a}, \theta)\pi_{SK}(\theta).
    \end{equation}

\subsection{Likelihood asymptotics}
\label{subapp:C2}
    In this appendix, we will study the asymptotics of the likelihood defined in (\ref{eq:SKr:likelihood-theta}) when $\beta\to\infty$.
    Let us first consider the substitution $u=(x-\log\hat a_j)/\sqrt{\sigma_\eps^2+\rho\beta^2}$ to express the likelihood as
    \begin{align*}
        \ell_J(\mbf z | \mbf{\hat a}, \theta) &= \prod_{j=1}^J\int_\RR f_{j}^\beta(u)du,\\
        f_{j}^\beta(u) &= \prod_{i=1}^{I_j}\Phi(h_j^\beta(u))^{z_{ij}}(1-\Phi(h_j^\beta(u)))^{1-z_{ij}} \frac{e^{-\frac{u^2}{2}}}{\sqrt{2\pi}} ,
    \end{align*}
    with 
        $$h_j^\beta(u) = \frac{(u+\log\hat a_j)\sqrt{\sigma_\eps^2+\rho\beta^2}-\log\alpha}{\beta\sqrt{1-\rho}}.
        $$
    This way, remembering that $0\leq\Phi(1-\Phi)\leq1$, an upper bound $u\mapsto e^{-u^2/2}/\sqrt{2\pi}$ can be found for $f_{j}^\beta$ for any $\beta,u$. It converges when $\beta \to +\infty$ as follows:
        \begin{equation*}
            \lim_{\beta\rightarrow\infty} f_{j}^\beta(u)%
                = \prod_{i=1}^{I_J}%
                \Psi^{z_{ij}}\left(\frac{(u+\log\hat a_j)\sqrt{\rho}}{\sqrt{1-\rho}}\right)
                \frac{e^{-\frac{u^2}{2}}}{\sqrt{2\pi}}.
        \end{equation*}
    This gives the following limit for the likelihood:
        \begin{align}\label{eq:SKr:limitlik}
            &\lim_{\beta\rightarrow\infty}\ell_J(\mbf z | \mbf{\hat a}, \theta)\\
                &= \prod_{j=1}^J\int_\RR \prod_{i=1}^{I_j}%
                \Psi^{z_{ij}}\left(\frac{(u+\log\hat a_j)\sqrt{\rho}}{\sqrt{1-\rho}}\right)
                \frac{e^{-\frac{u^2}{2}}}{\sqrt{2\pi}}du ,\nonumber
        \end{align}
    which is a positive quantity.

    \subsection{Posterior asymptotics}
\label{subapp:C3}
    By combining Eqs.~(\ref{eq:SKr:limitlik}), (\ref{eq:SKr:post}) and (\ref{eq:SKr:SKprior}) we obtain the posterior asymptotics
    \begin{equation}
            p(\theta|\mbf z, \mbf{\hat a}) \equi{\beta\rightarrow{\infty}} \frac{C}{\beta},
        \end{equation}
 with $C$ being a positive constant.
    This makes the posterior improper w.r.t. $\beta$, with the same convergence rate as the one derived in our framework.

\section{Supplementary material}

Supplementary material related to this article can be found online
at \href{https://ars.els-cdn.com/content/image/1-s2.0-S0266892024000444-mmc1.pdf}{https://ars.els-cdn.com/content/image/1-s2.0-S0266892024000444-mmc1.pdf}.

\bibliographystyle{elsarticle-num-names} 
\bibliography{biblio}

\end{document}